\documentclass[sn-mathphys,Numbered]{sn-jnl}

\usepackage{graphicx}%
\usepackage{subfig}
\usepackage{multirow}%
\usepackage{amsmath,amssymb,amsfonts}%
\usepackage{amsthm}%
\usepackage{mathrsfs}%
\usepackage[title]{appendix}%
\usepackage{xcolor}%
\usepackage{textcomp}%
\usepackage{manyfoot}%
\usepackage{booktabs}%
\usepackage{algorithm}%
\usepackage{algorithmicx}%
\usepackage{algpseudocode}%
\usepackage{listings}%
\usepackage{caption}
\usepackage{float}
\usepackage{makecell}
\usepackage{natbib}
\usepackage[section]{placeins}

\newtheorem{definition}{\textbf{Definition}}
\newtheorem{theorem}{\textbf{Theorem}}
\newtheorem{example}{\textbf{Example}}%

\begin{document}

\title[Article Title]{Efficient \textit{k}-step Weighted Reachability Query Processing Algorithms}

\author[1]{Congquan Mei}
\author[1]{Lian Chen}
\author*[1]{Junfeng Zhou}\email{zhoujf@dhu.edu.cn}
\author[1]{Ming Du}
\author[1]{Sheng Yu}
\author[2]{Xian Tang}
\author[3]{Ziyang Chen}
\affil[1]{School of Computer Science and Technology, Donghua University, Shanghai, 201620,  China}

\affil[2]{School of Electronic and Electrical Engineering, Shanghai University of Engineering Science, Shanghai, 201620, China}

\affil[3]{Shanghai Lixin University of Accounting and Finance, Shanghai, 201620, China}

\abstract{Given a data graph \textit{G}, a source point \textit{u} and a target point \textit{v} of a reachability query, the reachability query is used to answer whether there exists a path from \textit{u} to \textit{v} in \textit{G}. Reachability query processing is one of the fundamental operations in graph data management, which is widely used in biological networks, communication networks, and social networks to assist data analysis. The data graphs in practical applications usually contain information such as quantization weights associated with the structural relationships, in addition to the structural relationships between vertices. Thus, in addition to the traditional reachability relationships, users may want to further understand whether such reachability relationships satisfy specific constraints. In this paper, we study the problem of efficiently processing \textit{k}-step reachability queries with weighted constraints in weighted graphs. The \textit{k}-step weighted reachability query questions are used to answer the question of whether there exists a path from a source point \textit{u} to a goal point \textit{v} in a given weighted graph. If it exists, the path needs to satisfy 1) all edges in the path satisfy the given weight constraints, and 2) the length of the path does not exceed the given distance threshold \textit{k}. To address the problem, firstly, WKRI indexes supporting \textit{k}-step weighted reachability query processing and index construction methods based on efficient pruning strategies are proposed. Secondly, the idea of constructing  indexes based on part of the points is proposed to reduce the size of the indexes and two optimized indexes are designed and implemented based on the vertex coverage set to design and implement two optimized indexes GWKRI and LWKRI. Finally, experiments are conducted on several real datasets. The experimental results verify the efficiency of the method proposed in this paper in answering \textit{k}-step weighted reachability queries.}


\maketitle

\section{Introduction}\label{sec1}

Reachability query processing is one of the fundamental problems in graph data management\cite{1duong2021efficient,2sengupta2019arrow,3zhou2017dag,4chen2008efficient,5chen2011decomposing,6trissl2007fast,7wang2006dual,8peng2020answering,9chen2021efficiently}   and is used to detect the existence of a connected path between two points in a graph. Reachability query processing is widely used in application areas such as social networks \cite{10choudhary2015survey}, road networks \cite{11cheng2012efficient}, and PPI networks\cite{12jin2008efficiently}. For example, in social networks, a reachability query can be used to determine whether there is a specific interaction between two given people; in road networks, a reachability query can answer whether there exists a route to reach B from A. In practical applications, in addition to reflecting the structural relationships between vertices, data graphs usually have information such as quantization weights related to the structural relationships. For example, in social networks, the edges representing user interaction relations usually contain information about the interaction time; in communication networks, the edges representing links are usually accompanied by information about the link bandwidth, and so on. Therefore, in addition to considering the graph structure, users usually pay more attention to reachability queries that satisfy specific constraints.

In existing work, Wen et el.\cite{13wen2020efficiently,14qiao2013computing,15peng2023efficiently} investigates weight-constrained reachability queries, which are used to determine whether there exists a connected path between two given points, where the values of the attributes on each edge satisfy the given constraints. For example, in order to guarantee the end-to-end quality of service in a communication network, it is necessary to ensure that all links need to satisfy a minimum bandwidth requirement over the possible transmission paths \cite{16gurukar2015commit}. In addition, \textit{k}-step reachability queries with distance constraints have been studied in the Yano et el.\cite{17yano2013fast,18yildirim2010grail,19cheng2012k,20peng2022answering}. For example, in a road network, a path with a distance less than a certain threshold between two places is recommended for a user when he or she is traveling. When using social networks to assist in infectious disease research, if there exists a path with a length less than a certain threshold between two people, i.e., the number of people in the transmission chain between two people is less than the threshold, then these two people may have a close or sub-close relationship.

In addition to the single-constraint reachability query problem studied in existing work, query requirements in real-world applications may require that both weight and distance constraints be considered. For example, in wireless or sensor networks, the probability of message reception decreases at an exponential rate with increasing transmission distance due to the fact that a message in transmission may be lost at any node\cite{19cheng2012k}. Therefore, data transmission in wireless or sensor networks, in addition to the requirement that the link bandwidth is greater than a user-defined threshold, the path length of data transmission should also satisfy a threshold set in advance; as another example, in the case of using social networks to assist in the study of infectious diseases, if there is a need to find the close and sub-close connections, in addition to the need to detect the number of people on the propagation chain, i.e., the propagation distance between two people, to see if it is less than a given threshold, it is also necessary to further test the time in the chain of transmission to see if it is close to the date of diagnosis of the infected person.

Existing methods cannot efficiently answer reachability queries when both weight and distance constraints are considered. The reason is reflected in two aspects: on the one hand, the indexes used by the existing methods for weight constraints \cite{13wen2020efficiently,14qiao2013computing,15peng2023efficiently} cannot reflect the path length, and the existing methods supporting distance \cite{17yano2013fast,18yildirim2010grail,19cheng2012k,20peng2022answering} do not contain information about the weights on the edges in the paths; on the other hand, the indexes of the two types of methods do not have specific data mapping relationships between them, and it is not possible to answer the queries that have both weight constraints and distance constraints by a simple combination of indexes of the query.

To solve these  problems, this paper investigates Weight-Constraint with \textit{K}-step Reachability (WCKR) query based on weight and distance constraints and its efficient processing. The basic idea is to design index structures that can cover the length and weight information of all the paths in the graph, based on which the results of WCKR queries are returned quickly. The specific contributions of this paper are as follows:

\begin{quote}
\setlength{\leftskip}{1em}
(1) We propose an efficient index, termed the WKRI (Weighted \textit{k}-step Reachability Index), along with a novel index construction algorithm tailored for addressing WCKR queries. This algorithm employs two advanced pruning strategies to generate an index that eliminates redundant information among all vertex pairs, thus enhancing both storage efficiency and query performance.
\end{quote}
\begin{quote}
\setlength{\leftskip}{1em}
(2) We introduce global and local indexing strategies to optimize index size and accelerate query processing. The global strategy constructs an index encompassing all vertices based on minimum point coverage, while the local strategy focuses on vertices within the minimum point coverage set. These strategies significantly improve query efficiency and reduce the overall index size.
\end{quote}
\begin{quote}
\setlength{\leftskip}{1em}
(3) Comprehensive experiments on various real-world datasets demonstrate the effectiveness of our proposed methods. The experimental results validate the efficiency and performance improvements of our approach from multiple perspectives.
\end{quote}

The rest of this paper is organized as follows. Section 2 discusses the basics as well as related work, Section 3 proposes WKRI index for WCKR queries and gives the index construction and query processing methods, Section 4 proposes two optimized indexes GWKRI and LWKRI based on the idea of minimum point coverage. Finally, experimental results are shown in Section 5 and the full paper is summarized in Section 6.

\section{Preliminaries}\label{sec2}

An undirected weighted graph \textit{G} can be represented by a quadruple as $\textit{G} 
=\left(V,E,\mathrm{\Sigma},w\right)$
 where \textit{V} and \textit{E} denote the set of vertices and the set of edges, $\mathrm{\Sigma}$ denotes the set of weights on all the edges in \textit{E}, and \textit{w} is the mapping function of edges to weights, with \textit{w}(\textit{e}) denoting the weights on edge \textit{e}. For any vertex $\textit{u}\in \textit{V}$, the set of neighboring vertices of \textit{u} is denoted as $N\left(u\right)=\left\{v\middle|\left(u,v\right)\in E\right\}$. Obviously, the degree of \textit{u} $d\left(u\right)=\left|N\left(u\right)\right|$. Any path \textit{P} connecting two vertices \textit{u} and \textit{v} can be represented as a sequence of vertices $\left(u=v_0,\ v_1,\ \ldots,{\ v}_i=v\right)$, where$\left(v_j,\ v_j+1\right)\in\ E\left(0\ \le\ j\le\ i-1\right)$. The length of \textit{P} is $len\left(P\right)= i+1$. Assuming that $P_{min}$ denotes the shortest path among all the paths between \textit{u}, \textit{v}, the shortest distance between \textit{u}, \textit{v} can be represented by  $P_{min}$, i.e., $dis\left(u,\ v\right)=\ len\left(P_{min}\right)$.

\begin{definition}
(Weight Constraint Reachable) Given a weight constraint \textit{c}, any two vertices $u,v \in V$. The point \textit{u} is said to be weight-constraint-reachable to \textit{v} if there exists a path \textit{P} from \textit{u} to \textit{v}, and the weights of all edges in \textit{P} satisfy \textit{c}.
\end{definition}

\begin{definition}
(\textit{k}-step reachable) Given a distance threshold \textit{k}, $\forall\ u,v \in\ V$. The point \textit{u} is said to be \textit{k}-step reachable to \textit{v} if there exists a path \textit{P} from \textit{u} to \textit{v} that satisfies $dis\left (P\right) \le k$.
\end{definition}

\begin{figure}[h]
    \centering
    \includegraphics[width=0.5\linewidth]{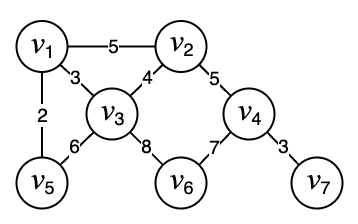}
    \caption{Weighted undirected graph\textit{ G}}
    \label{fig1}
\end{figure}

\textbf{Problem Statement. } Given an undirected weighted graph \textit{G} and a query\textit{ Q}(\textit{u, v, c, k}), where \textit{u},\textit{v} are the two vertices of the query, \textit{c} is a constraint, which can be either a semi-bounded interval $\le x$, $\geq y$, or a bounded interval [\textit{x}, \textit{y}], and \textit{k} is a distance threshold. The weight-constrained \textit{k}-step reachability query \textit{Q}(\textit{u, v, c, k}) is used to answer the question of whether there exists a path \textit{P} between vertices \textit{u}, \textit{v} where the weights on all edges in \textit{P} satisfy the weight constraints \textit{c} and $len\left(P\right)\le k$.
\medskip
\begin{example}
In the undirected weighted graph \textit{G} of Fig. \ref{fig1}, we have $v_5$ 3-step reaches $v_2$ in the bounded interval [\textit{2}, \textit{4}], since there exists a path {($v_5$, $v_1$, 2), ($v_1$, $v_3$, 3), ($v_3$, $v_2$, 4)} from $v_5$ to $v_2$ with weight constraints [\textit{2}, \textit{4}] in the graph \textit{G}.

Additionally, consider a query with a semi-bounded interval $\leq 3$. In this case, $v_5$ 2-step reaches $v_3$ because there exists a path {($v_5$, $v_1$, 2), ($v_1$, $v_3$, 3)} where the edge weight $2$ and $3$ is less than or equal to $3$, satisfying the weight constraint $\leq 3$.
\end{example}
\subsection{Related Work}
\subsubsection{Weight-Constrained Reachability Query Algorithm}
\citet{13wen2020efficiently} studied Efficient processing of Span-Reachability queries on temporal graphs . Given a temporal graph \textit{G}, two vertices \textit{u} and \textit{v}, and a time interval $\left[t_s,t_e\right]$ as a constraint, a Span-Reachability query on the temporal graph is used to detect whether \textit{u} can reach \textit{v} through an edge whose moment value falls within the interval $\left[t_s\ ,t_e\right]$. If the moment information on an edge is considered as a weight, this type of query processing problem belongs to the weight-constrained-reachability query problem. The authors propose a TILL index to support efficient processing of  Span-Reachability queries on temporal graphs.

\citet{14qiao2013computing} proposes a binary index tree based on edge weights to address the weight-constrained reachability (WCR) problem in ordinary graphs. The construction of the edge-indexed binary tree, for instance with a constraint $\le y$, involves two main steps. First, the input weighted graph is used to construct a minimum spanning tree \textit{T}. Then, the edge-indexed binary tree is built based on \textit{T}.The construction process of the edge-indexed binary tree is as follows: (1) Identify the edge \textit{e} with the largest weight in \textit{T} and designate it as the root node of the edge-indexed binary tree. (2) Remove edge \textit{e} from \textit{T}, which splits \textit{T} into two subtrees, $T_1$ and $T_2$. These steps are repeated until all edges and vertices of \textit{T} have been removed. Finally, all vertices from \textit{T} are inserted into their appropriate positions within the edge-indexed binary tree in a bottom-up manner.For a given query, a vertex \textit{u} is weight-constrained reachable to another vertex \textit{v} if the weight of the lowest common ancestor of \textit{u} and \textit{v} in the edge-indexed binary tree satisfies the constraint $c$.

\citet{15peng2023efficiently} studied the shortest distance query problem with quality constraints in large graphs. For real-life graphs, if the quality values on the edges (e.g., bandwidth in a routing network) are regarded as weights, when given a quality constraint, the edges that satisfy the constraint are valid edges. The article focuses on investigating whether it is possible to find a shortest path along a valid edge for two vertices \textit{s} and \textit{t} given in a graph. Although the problem is a shortest distance query problem, the core idea is still to construct a 2-hop index based on weights and perform a fast shortest distance query based on it, it is similar to the weight-constrained reachability problem.

The above mentioned \citet{13wen2020efficiently,15peng2023efficiently} constructed 2-hop indexes to answer the reachability as well as the shortest distance problem, while \citet{14qiao2013computing} utilized tree indexes to answer the reachability problem. However, none of these three methods are applicable to solve the WCKR problem because none of the indexes constructed by these three methods take into account the constraints of path length.
\subsubsection{\textit{K}-step Reachability Query Algorithm}
For the \textit{k}-step reachability query problem, traditional solutions include the PLL \cite{17yano2013fast} algorithm as well as the K-Reach \cite{19cheng2012k} algorithm.

PLL can answer a \textit{k}-step reachability query by comparing the length of the shortest path between two points with a given distance threshold by recording the length of the shortest path in the index label entries. The index label of each vertex \textit{u} in the PLL index can be expressed as $L\left(u\right)=\left\{\left({hop}_1,\ {dis}_1\right ),\ \left({hop}_2,\ {dis}_2\right),\ \ \ldots\ ,\left({hop}_l,\ {dis}_l\right)\right\}$, where $({hop}_i,\ {dis}_i)$ denotes that the length of the shortest path from vertex \textit{u} to vertex ${hop}_i$ is ${dis}_i (1\le i \le l)$. When processing the query \textit{Q}(\textit{u}, \textit{v}, \textit{k}), first determine whether there exists a common vertex ${hop}_i$  in \textit{L}(\textit{u}) and \textit{L}(\textit{v}), and if so, further determine whether the sum of the shortest distance from \textit{u} to ${hop}_i$  and that from \textit{v} to ${hop}_i$  is$\le k$. If the condition is satisfied, then it means that \textit{u} reaches the vertex \textit{v} in step \textit{k}. To build the index, a BFS traversal is required from each hop point ${hop}_i$ . The PLL algorithm constructs indexes for the \textit{k}-step reachability query problem on ordinary graphs. The indexes contain only the length information of the paths, and thus cannot answer the  \textit{k}-step reachability query with weight constraints.

K-Reach constructs a \textit{k}-step passing closure covering all vertices in the set based on the minimum vertex cover and records the distance between each vertex and the vertices in its passing closure. When answering a \textit{k}-step reachable query, the query vertices are processed separately based on whether they belong to the minimum vertex coverage or not. As the value of \textit{k} increases, the size of the passing closure for each vertex rises exponentially, thus the algorithm has a large index size. On the other hand, the authors have also proposed an optimization method that trades time for space to reduce the index size. However, like PLL, K-Reach maintains distances in the index and has no information about the weights on the edges in the path, and thus cannot handle \textit{k}-step reachable queries with weight constraints.

\citet{20peng2022answering} proposed a novel approach to answer the \textit{k}-step reach problem with label constraints based on the above two algorithms. In this paper, the authors proposed  a search algorithm based on upper and lower bounds to answer the K-Reach query problem. It is the first algorithm that uses 2-hop indexes as upper and lower bounds to answer K-Reach queries. Although the method combines the 2-hop index and the K-Reach algorithm, its main idea is to solve the label-constrained reachability problem, and thus it cannot deal with weight-constrained k-step reachability queries, but the idea is still worthwhile.
\subsubsection{Other Reachability Query Algorithm}
In addition to the above reachability query processing methods for weights constraints and k-steps constraints, researchers have also conducted in-depth studies on other types of reachability queries, including parallel processing methods for unconstrained reachability queries \cite{21li2019scaling,22li2020scaling}, sink-graph reachability query processing methods based on graph parsimony \cite{23dietrich2021efficient}, label-constrained reachability query processing methods based on dynamic graphs \cite{24chen2022dlcr}, and label-constrained reachability query processing methods based on distributed computing \cite{25zeng2022distributed}.  These methods effectively extend the application scope of reachability queries.

\section{WKRI Algorithm}\label{sec3}
Given a weighted undirected graph $\textit{G} 
=\left(V,E,\mathrm{\Sigma},w\right)$ and a query \textit{Q}(\textit{u}, \textit{v}, \textit{c}, \textit{k}), the most straightforward way to solve the WCKR query problem is to perform a BFS/DFS traversal \cite{26bundy1986catalogue,27veloso2014reachability}. Specifically, for each query \textit{Q}, we performed a BFS traversal from vertex \textit{u}. In the traversal, only edges that satisfy the weight constraint \textit{c} are traversed, and vertex \textit{u} is reachable to vertex \textit{v} if vertex \textit{v} can be reached and the path length also satisfies the distance constraint \textit{k}.

Although BFS or DFS traversal can solve the WCKR query problem, the query process may access all the vertices in the graph, and thus the query is inefficient and poorly scalable. In this section, we propose a novel algorithm——WKRI, which can efficiently solve the WCKR query problem.

\subsection{WKRI Index}
The 2-hop index records the shortest path length, which is used to solve the shortest distance query problem and can also be used to answer the \textit{k}-step reachability query problem. For example,  for the weighted undirected graph \textit{G} of Fig. 1, the corresponding 2-hop distance index is shown in Table~\ref{tab1}. For each vertex \textit{u}, each row in the table denotes the label \textit{L}(\textit{u}) of that vertex, and each item \textit{I} = (\textit{w},\textit{ k}) in \textit{L}(\textit{u}) is a binary indicating that the shortest distance from vertex \textit{w} to vertex \textit{u} is \textit{k}. Answering whether \textit{u} can \textit{k}-step-reachable to \textit{v} based on this index can be described simply as follows: given a query \textit{Q} = (\textit{u}, \textit{v}, \textit{k}), if there exists $\left(w_i,k_1\right)\in\ L\left(u\right)$, $\left(w_j,k_2\right)\in\ L\left(v\right)$ that satisfies $w_i=\ w_j$ and $k_1\ +\ k_2\ \le\ k$, then the result of the query \textit{Q} will be true, and false otherwise.

\begin{table}[htbp]
  \centering
  \caption{2-hop distance index}
  \renewcommand{\arraystretch}{1.5}
    \begin{tabular}{c|cccc}
    \hline
    ID    & \multicolumn{4}{c}{Label} \\
    \hline
    \textit{L}(\textit{v}\textsubscript{1})  & (\textit{v}\textsubscript{3},1) & (\textit{v}\textsubscript{4},2) & (\textit{v}\textsubscript{2},1) & \multicolumn{1}{c}{(\textit{v}\textsubscript{1},0)} \\
    \hline
    \textit{L}(\textit{v}\textsubscript{2})  & (\textit{v}\textsubscript{3},1) & (\textit{v}\textsubscript{4},1) & (\textit{v}\textsubscript{2},0) &  \\
    \hline
    \textit{L}(\textit{v}\textsubscript{3})  & (\textit{v}\textsubscript{3},0) & \multicolumn{1}{c}{} & \multicolumn{1}{c}{} &  \\
    \hline
    \textit{L}(\textit{v}\textsubscript{4})  & (\textit{v}\textsubscript{3},2) & (\textit{v}\textsubscript{4},0) & \multicolumn{1}{c}{} &  \\
    \hline
    \textit{L}(\textit{v}\textsubscript{5})  & (\textit{v}\textsubscript{3},1) & (\textit{v}\textsubscript{1},1) & (\textit{v}\textsubscript{5},0) &  \\
    \hline
    \textit{L}(\textit{v}\textsubscript{6})  & (\textit{v}\textsubscript{3},1) & (\textit{v}\textsubscript{4},1) & (\textit{v}\textsubscript{6},0) &  \\
    \hline
    \textit{L}(\textit{v}\textsubscript{7})  & (\textit{v}\textsubscript{3},3) & (\textit{v}\textsubscript{4},1) & (\textit{v}\textsubscript{7},0) &  \\
    \bottomrule
    \end{tabular}%
    \begin{center}
\parbox{0.8\textwidth}{The vertex processing order is $v_3, v_4, v_2, v_1, v_6, v_5, v_7$.
}
\end{center}
  \label{tab1}%
\end{table}%
As can be seen from Table~\ref{tab1}, the 2-hop distance index only retains information about the length of the shortest path between two points. If the WCKR query is answered by adding the weights on the edges in the shortest paths to each item in this index, the information about the weights of the non-shortest paths will be missed, resulting in inaccurate query results. For example, a query through Table~\ref{tab1} shows that the path from $v_1$ to $v_3$ is only $v_1 \rightarrow v_3$, which is also the shortest path from $v_1$ to $v_3$. However, according to Fig. 1, it is known that there are other non-shortest paths from $v_1$ to $v_3$, such as $v_1 \rightarrow v_5 \rightarrow v_3$, $v_1 \rightarrow v_2 \rightarrow v_3$, and so on. If a given query needs to detect whether there exists a path between vertices $v_1$ to $v_3$ with a length not exceeding 2 and a weight between 4 and 6, it can be seen from Figure 1 that the path that satisfies the condition is $v_1 \rightarrow v_2 \rightarrow v_3$. However, only the information of the shortest paths is recorded in Table~\ref{tab1}. For this query, the path that satisfies the condition is not recorded in Table~\ref{tab1}. If weight information is added only to the index entries in Table~\ref{tab1}, then weight information on non-shortest paths such as $v_1 \rightarrow v_5 \rightarrow v_3$, $v_1 \rightarrow v_2 \rightarrow v_3$, etc., will be missed. Therefore, the index used to answer the WCKR query needs to cover the weight information corresponding to all paths. Taking Fig. 1 as an example, BFS traversal is used to visit all the vertices, and all the paths traveled during the traversal and the weight information are recorded. Part of the constructed index (contains only three-hop vertices $v_1, v_2, v_3$) is shown in Table~\ref{tab2}. For each vertex \textit{u}, each row in the table represents the path label \textit{L}(\textit{u}) of that vertex, where each label item $I=\left(v,w_s,w_e,k\right)$ is a quaternion representing a path \textit{P} from \textit{u} to \textit{v}, which has the range of the weights on the edges of the path \textit{P} as $\left[w_s,w_e\right]$, and the distance of the path is \textit{k}. Although this index can correctly answer the WCKR query, the index construction method based on BFS traversal has both time and space complexity as high as $O\left(\left|V\right|^2\right)$, which leads to a large index size and low query efficiency in practice. For this reason, efficient redundant label identification methods need to be designed to compress the index size and improve the query efficiency.

\begin{definition}
(Path Tuple) For any vertex \textit{u}, each label item $I=\left(v,w_s,w_e,k\right)$ in \textit{L}(\textit{u}) denotes a reachable path \textit{P} from \textit{u} to \textit{v}, which can be expressed in terms of a path tuple as $T=\left(u,v,w_s,w_e,k\right)$.
\end{definition}

Each label item in the index can unidirectionally generate a path tuple, indicating that the path corresponding to the path tuple is reachable.

\begin{definition}
(Path Domination) Given two path tuples $t_1=\left(u,\ v,w_s,w_e,\ k\right)$ and $t_2=\left(u,v,w_s^1,w_e^1,k_1\right)$, a path tuple $t_1$ is said to dominate $t_2$ if $\left[w_s,w_e\right]\subseteq\left[w_s^1,\ {\ w}_e^1\right]$ and $k \le\ k_1$.
\end{definition}

For the path tuples $t_1$ and $t_2$ in Definition 4, we find that if $t_1$ denotes that vertex \textit{u} to vertex \textit{v} is reachable in \textit{k} steps within $\left[w_s,w_e\right]$, then when $t_1$ and $t_2$ satisfy the governing condition that $\left[w_s,w_e\right]\subseteq\left[w_s^1,{w}_e ^1\right]$ and $k\le k_1$, vertex \textit{u} to vertex \textit{v} must also be reachable in $k_1$ steps within $\left[w_s^1,w_e^1\right]$, so the reachable path information represented by the tuple $t_2$ is in fact redundant information, and when the path tuple is generated, we can express its reachable information with $t_1$ instead, so $t_2$ is in fact is a redundant tuple, and the label item that generates this tuple is also a redundant label item. That is, when the path tuple $t_1$ dominates $t_2$, $t_2$ is a redundant path tuple, and its corresponding label item is also a redundant label item.

\begin{theorem}
(Redundancy detection rules for label items at the same hop point) If there is a label item $I=\left(v,w_s,w_e,k\right)$ in the label \textit{L}(\textit{u}) of \textit{u}, $I_1$ is redundant when the label item $I_1\ =\left(v,w_s^1,\ {w}_e^1,k_1\right)$ satisfies $[w_s,w_e] \subseteq [w_s^1,{w}_e ^1]$ and $k\le k_1$.
\end{theorem}
\begin{proof}
The label item $\textit{I}_1$ can be represented as the path tuple $\left(u,v,w_s^1,w_e^1,k_1\right)$, and the labeling item \textit{\textbf{I}} can be represented as the path tuple $\left(u,v,w_s,w_e,k\right)$.  According to the known conditions, we know that $\left(u,v,w_s,w_e,k\right)$ dominates $\left(u,v,w_s^1,w_e^1,k_1\right)$. By inference from Definition 4, $I_1$ is a redundant label item.
\end{proof}

\begin{theorem}
(Redundancy detection rules for label items at different hop points) Suppose that after the hop vertex \textit{x} is processed, there exists a label item $I_1=\left(x,w_s^1,w_e^1,k_1\right)$ in the label \textit{L}(\textit{u}) of \textit{u}, and a label item $I_2=\left(x,w_s^2,w_e^2,k_2\right)$ in the label \textit{L}(\textit{v}) of \textit{v}. When dealing with hop point \textit{v}, if the labeling term $I_3=\left(v,w_s^3,w_e^3,k_3\right)$ is generated at point \textit{u}, $I_3$ satisfies $\left[\min{\left(w_s^1,{\ w}_s^2\right)},\max{\left(w_e^1,{\ w}_e^2\right)}\right]\subseteq\left[w_s^3,{\ w}_e^3\right]$ and $k_1+k_2\le k_3$, then $I_3$ is a redundant label item.
\end{theorem}

\begin{proof}
The positional relations of vertices \textit{u}, \textit{v}, \textit{x} are shown in Fig.2. $\textit{I}_1$ can be represented as path tuple $(u,x,w_s^1,w_e^1,k_1)$, $I_2$ can be represented as path tuple $(v,x,w_s^2,w_e^2,k_2)$ , and $I_3$ can be represented as path tuple $(u,v,w_s^3,w_e^3,k_3)$. Since $I_1$ and $I_2$ represent the paths of $u\rightarrow x$ and $x\rightarrow v$, respectively, $I_1$ and $I_2$ can be merged to represent the paths of $u\rightarrow v$. And the path tuples can be combined as $(u,v,min(w_s^1,w_s^2),max(w_e^1,w_e^2),k_1+k_2)$. Since $[min(w_s^1,{w}_s^2),max(w_e^1,{w}_e^2)] \subseteq [w_s^3, w_e^3]$ and $k_1+k_2\le k_3$, $(u,v,w_s^3,w_e^3,k_3)$ is dominated. It is a redundant path tuple, and its corresponding label, $I_3$, is redundant.    
\end{proof}

\begin{figure}[h]
    \centering
    \includegraphics[width=0.5\linewidth]{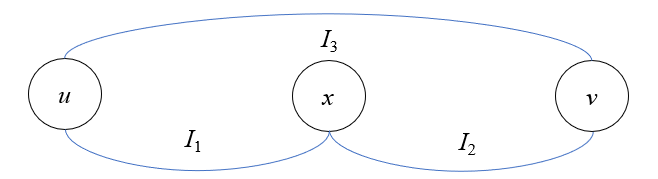}
    \caption{The positional relationships between \textit{u}, \textit{v} and \textit{x} }
    \label{fig2}
\end{figure}

\begin{example}
Theorem 1 and Theorem 2 can be used to detect whether the labels in Table \ref{tab2} are redundant or not, and then compress the index size by removing the redundant labels. For example, for \textit{L}(\textit{v\textsubscript{4}}) of Table \ref{tab2}, since the reachable tuples corresponding to ($v_3$, 4, 5, 2) dominate the reachable tuples corresponding to ($v_3$, 3, 5, 3) and ($v_3$, 2, 6, 4), according to Theorem 1, ($v_3$, 3, 5, 3) and ($v_3$, 2, 6, 4) are redundant labels. Further, consider the label item ($v_4$, 4, 6, 3) in L(\textit{v\textsubscript{5}}). Since there are ($v_3$, 4, 5, 2) in L(\textit{v\textsubscript{4}}) and ($v_3$, 6, 6, 1) in L(\textit{v\textsubscript{5}}), and since $\left[min\left(4,6\right),max\left(5,6\right)\right]\subseteq\left[4,6\right]$ and $2+1\\\le3$, according to Theorem 2, ($v_4$, 4, 6, 3) in L(\textit{v\textsubscript{5}}) is a redundant labeling term. Similarly, Theorem 1 and Theorem 2 can be used to detect whether other label items are redundant, and finally the index without redundant label items is obtained.
\end{example}
\begin{table}[htbp]
  \centering
  \caption{Part of index of all paths}
    \begin{tabular}{cccc}
    \toprule
    Label & $v_1$ & $v_2$ & $v_3$ \\
    \midrule
    \textit{L}($v_1$) & ($v_1$,0,0,0) & \multicolumn{1}{c}{} & \multicolumn{1}{c}{} \\
    \midrule
    \multicolumn{1}{c}{\multirow{2}[2]{*}{\textit{L}($v_2$)}} & ($v_1$,5,5,1) ($v_1$,2,6,3) & \multirow{2}[2]{*}{($v_2$,0,0,0)} & \multicolumn{1}{c}{\multirow{2}[2]{*}{}} \\
          & ($v_1$,3,4,2) ($v_1$,2,8,5) ($v_1$,3,8,4) & \multicolumn{1}{c}{} & \multicolumn{1}{c}{} \\
    \midrule
    \multicolumn{1}{c}{\multirow{2}[2]{*}{\textit{L}($v_3$)}} & ($v_1$,3,3,1) ($v_1$,5,8,4) & ($v_2$,4,4,1)($v_2$,2,6,3) & \multirow{2}[2]{*}{($v_3$,0,0,0)} \\
          & ($v_1$,2,6,2) ($v_1$,4,5,2) & ($v_2$,3,5,2) ($v_2$,5,8,3) & \multicolumn{1}{c}{} \\
    \midrule
    \multicolumn{1}{c}{\multirow{2}[2]{*}{\textit{L}($v_4$)}} & ($v_1$,5,5,2) ($v_1$,3,8,3) & ($v_2$,5,5,1) ($v_2$,3,8,4) & ($v_3$,4,5,2) ($v_3$,7,8,2) \\
          & ($v_1$,2,6,4) ($v_1$,3,5,3) ($v_1$,2,8,4) & ($v_2$,4,8,3) ($v_2$,2,8,5) & ($v_3$,3,5,3) ($v_3$,2,6,4) \\
    \midrule
    \multicolumn{1}{c}{\multirow{2}[2]{*}{\textit{L}($v_5$)}} & ($v_1$,2,2,1) ($v_1$,4,6,3) & ($v_2$,2,5,2) ($v_2$,3,6,3) & ($v_3$,6,6,1) ($v_3$,2,3,2) \\
          & ($v_1$,3,6,2) ($v_1$,5,8,5) & ($v_2$,2,8,5) ($v_2$,4,6,2) ($v_2$,5,8,4) & ($v_3$,2,5,3) ($v_3$,2,8,5) \\
    \midrule
    \multicolumn{1}{c}{\multirow{2}[2]{*}{\textit{L}($v_6$)}} & ($v_1$,3,8,2) ($v_1$,2,8,3) ($v_1$,5,7,3) & ($v_2$,3,8,3) ($v_2$,4,8,2) & ($v_3$,8,8,1) ($v_3$,4,7,3) \\
          & ($v_1$,2,7,5) ($v_1$,3,7,4) ($v_1$,4,8,3) & ($v_2$,2,8,4) ($v_2$,5,7,2) & ($v_3$,3,7,4)($v_3$,2,7,5) \\
    \midrule
    \multicolumn{1}{c}{\multirow{2}[2]{*}{\textit{L}($v_7$)}} & ($v_1$,3,5,3) ($v_1$,3,5,4) ($v_1$,2,8,5) & ($v_2$,3,5,2) ($v_2$,3,8,4) & ($v_3$,3,5,3) ($v_3$,3,8,3) \\
          & ($v_1$,3,8,5) ($v_1$,3,8,4) ($v_1$,2,6,5) & ($v_2$,2,8,6) ($v_2$,3,8,5) & ($v_3$,3,5,4) ($v_3$,2,6,5) \\
    \bottomrule
    \end{tabular}%
    \label{tab2}
\end{table}%

\subsection{Index Construction}
Based on Theorem 1 and Theorem 2, we propose the WKRI algorithm for constructing indexes without redundant labels to achieve efficient processing of k-step reachable queries on weighted graphs. The index construction process is as follows: first sort the vertices in descending order according to degree, then start from each vertex v, perform BFS traversal operation, and when traversing to the already processed hop vertices, end the traversal and perform pruning. At the same time, the redundancy detection rule based on Theorem 1 and Theorem 2 is used in the traversal process to prune the indexes containing weight intervals and path lengths. This pruning rule can be simply described as follows: during the traversal process starting from vertex \textit{v}, if vertex \textit{u} is encountered, the corresponding label item $I=(v,w_s,w_e,k)$ is generated. If \textbf{\textit{I} }is a redundant label item, the label item will be directly discarded, and the traversal will stop at vertex \textit{u}. If there are certain label items in the existing index, and these labeling items become redundant due to the existence of \textit{\textbf{I}}, the redundant label items that have already been inserted will be replaced by  \textit{\textbf{I}}. The specific flow of the WKRI algorithm is shown in Algorithm 1.

\begin{algorithm}
\caption{WKRI Index Construction Algorithm}
\begin{algorithmic}[1]
\renewcommand{\algorithmicrequire}{\textbf{Input:}}
\renewcommand{\algorithmicensure}{\textbf{Output:}}
\Require Graph $\textit{G} = (V, E, \Sigma, w)$
\Ensure Index label \textit{L}(\textit{u}) for each vertex \textit{u}
\State Sort all vertices by degree in descending order
\For{each vertex $\textit{v} \in \textit{V}$}
        \State Perform BFS starting from \textit{v}. For each lower-ranked vertex \textit{u} visited:
        \State $I \gets (v, w_s, w_e, k)$
        \If{isRedundant(\textit{I}) is false}
                \State Add \textit{I} to \textit{L}(\textit{u}) and expand \textit{u} in the BFS
        \Else
                \State Do not add \textit{I} to \textit{L}(\textit{u}) nor expand \textit{u}
        \EndIf
\EndFor
\State
\State \textbf{Function} \textit{isRedundant}(I)
\For{each $I_1 = (v, w_s^1, w_e^1, k_1) \in L(u)$}
        \If{$[w_s, w_e] \subseteq [w_s^1, w_e^1] \land k \leq k_1$}
                \State Delete $I_1$
        \ElsIf{$[w_s^1, w_e^1] \subseteq [w_s, w_e] \land k_1 \leq k$}
                \State \Return true
        \EndIf
\EndFor
\For{each $(x, w_s^2, w_e^2, k_2) \in L(u)$ and $(y, w_s^3, w_e^3, k_3) \in L(v)$}
        \If{$x = y \land [\min(w_s^2, w_s^3), \max(w_e^2, w_e^3)] \subseteq [w_s, w_e] \land k_2 + k_3 \leq k$}
                \State \Return true
        \EndIf
\EndFor
\State \Return false
\end{algorithmic}
\end{algorithm}

\begin{example}
    Consider the undirected weighted graph \textit{G} shown in Fig. 1, which is sorted according to degree and the vertex processing order is $O = \{v_3, v_4, v_2, v_1, v_6, v_5, v_7\}$. The construction processing of the index shown in Table \ref{tab3} is as follows: first, a BFS traversal will be performed from $v_3$ as the starting point, and a label item $I= (v_3, w_s, w_e, k)$ will be firstly generated. Then the function isRedundant() is called to check whether \textbf{I} is redundant or not. If \textbf{I} is not a redundant label item, then \textbf{I} will be inserted into the index label of $v$. According to Fig. 1, there are multiple paths between $v_3$ and $v_2$. Suppose that the label item $(v_3, 4, 4, 1)$ has already been deposited into $L(v_2)$, and when the path $v_3 \rightarrow v_1 \rightarrow v_2$ is accessed, its label term is $(v_3, 3, 5, 2)$. According to the pruning rule, this label term is a redundant label term and cannot be inserted into $L(v_2)$, at the same time, it is no longer traversed outward from $v_2$. The other vertices are treated similarly. Each row of Table \ref{tab3} represents the label $L(u)$ of vertex $u$, and each column corresponds to a vertex $v$, whose value represents the set of label items associated with $v$ in $L(u)$.
\end{example}

\begin{table}[htbp]
  \centering
  \caption{WKRI index based on WKRI algorithm}
    \begin{tabular}{cccccccc}
    \toprule
    Label    & \textit{v}\textsubscript{3} & \textit{v}\textsubscript{4} & \multicolumn{1}{c}{\textit{v}\textsubscript{2}} & \multicolumn{1}{c}{\textit{v}\textsubscript{1}} & \multicolumn{1}{c}{\textit{v}\textsubscript{6}} & \multicolumn{1}{c}{\textit{v}\textsubscript{5}} & \textit{v}\textsubscript{7} \\
    \midrule
    \multirow{3}[2]{*}{\textit{L}(\textit{v}\textsubscript{1})} & (\textit{v}\textsubscript{3},3,3,1) & \multirow{3}[2]{*}{(\textit{v}\textsubscript{4},5,5,2)} & \multicolumn{1}{c}{\multirow{3}[2]{*}{(\textit{v}\textsubscript{2},5,5,1)}} & \multicolumn{1}{c}{\multirow{3}[2]{*}{(\textit{v}\textsubscript{1},0,0,0)}} & \multirow{3}[2]{*}{} & \multirow{3}[2]{*}{} & \multicolumn{1}{c}{\multirow{3}[2]{*}{}} \\
    \multicolumn{1}{c}{} & (\textit{v}\textsubscript{3},4,5,2) & \multicolumn{1}{c}{} &       &       &       &       & \multicolumn{1}{c}{} \\
    \multicolumn{1}{c}{} & (\textit{v}\textsubscript{3},5,8,4) & \multicolumn{1}{c}{} &       &       &       &       & \multicolumn{1}{c}{} \\
    \midrule
    \multirow{2}[2]{*}{\textit{L}(\textit{v}\textsubscript{2})} & (\textit{v}\textsubscript{3},4,4,1) & \multirow{2}[2]{*}{(\textit{v}\textsubscript{4},5,5,1)} & \multicolumn{1}{c}{\multirow{2}[2]{*}{(\textit{v}\textsubscript{2},0,0,0)}} & \multirow{2}[2]{*}{} & \multirow{2}[2]{*}{} & \multirow{2}[2]{*}{} & \multicolumn{1}{c}{\multirow{2}[2]{*}{}} \\
    \multicolumn{1}{c}{} & (\textit{v}\textsubscript{3},5,8,3) & \multicolumn{1}{c}{} &       &       &       &       & \multicolumn{1}{c}{} \\
    \midrule
    \textit{L}(\textit{v}\textsubscript{3})  & (\textit{v}\textsubscript{3},0,0,0) & \multicolumn{1}{c}{} &       &       &       &       & \multicolumn{1}{c}{} \\
    \midrule
    \multirow{2}[2]{*}{\textit{L}(\textit{v}\textsubscript{4})} & (\textit{v}\textsubscript{3},4,5,2) & \multirow{2}[2]{*}{(\textit{v}\textsubscript{4},0,0,0)} & \multirow{2}[2]{*}{} & \multirow{2}[2]{*}{} & \multirow{2}[2]{*}{} & \multirow{2}[2]{*}{} & \multicolumn{1}{c}{\multirow{2}[2]{*}{}} \\
    \multicolumn{1}{c}{} & (\textit{v}\textsubscript{3},7,8,2) & \multicolumn{1}{c}{} &       &       &       &       & \multicolumn{1}{c}{} \\
    \midrule
    \multirow{2}[2]{*}{\textit{L}(\textit{v}\textsubscript{5})} & (\textit{v}\textsubscript{3},6,6,1) & \multirow{2}[2]{*}{(\textit{v}\textsubscript{4},2,5,3)} & \multicolumn{1}{c}{\multirow{2}[2]{*}{(\textit{v}\textsubscript{2},2,5,2)}} & \multicolumn{1}{c}{\multirow{2}[2]{*}{(\textit{v}\textsubscript{1},2,2,1)}} & \multirow{2}[2]{*}{} & \multicolumn{1}{c}{\multirow{2}[2]{*}{(\textit{v}\textsubscript{5},0,0,0)}} & \multicolumn{1}{c}{\multirow{2}[2]{*}{}} \\
    \multicolumn{1}{c}{} & (\textit{v}\textsubscript{3},2,3,2) & \multicolumn{1}{c}{} &       &       &       &       & \multicolumn{1}{c}{} \\
    \midrule
    \multirow{2}[2]{*}{\textit{L}(\textit{v}\textsubscript{6})} & (\textit{v}\textsubscript{3},8,8,1) & \multirow{2}[2]{*}{(\textit{v}\textsubscript{4},7,7,1)} & \multirow{2}[2]{*}{} & \multirow{2}[2]{*}{} & \multicolumn{1}{c}{\multirow{2}[2]{*}{(\textit{v}\textsubscript{6},0,0,0)}} & \multirow{2}[2]{*}{} & \multicolumn{1}{c}{\multirow{2}[2]{*}{}} \\
    \multicolumn{1}{c}{} & (\textit{v}\textsubscript{3},4,7,3) & \multicolumn{1}{c}{} &       &       &       &       & \multicolumn{1}{c}{} \\
    \midrule
    \textit{L}(\textit{v}\textsubscript{7})  & (\textit{v}\textsubscript{3},3,5,3) & (\textit{v}\textsubscript{4},3,3,1) &       &       &       &       & (\textit{v}\textsubscript{7},0,0,0) \\
    \bottomrule
    \end{tabular}%
  \label{tab3}%
\end{table}%

\begin{theorem}
    The WKRI algorithm constructs indexes that do not contain redundant path information.
\end{theorem}
\begin{proof}
    For any vertex \textit{u}, \textit{L}(\textit{u}) denotes the index labeling of \textit{u}. $I=(v,w_s,w_e,k)\in L(u)$ denotes the label itme. An index that does not contain redundant path information denotes that there must exist at least one query that cannot be answered by the index with the removal of any label item \textit{\textbf{I}}. Assuming that a WCKR query between two points can still be answered after removing $I=\ (v,w_s,w_e,k)\in L(u)$, then at least one of the following cases is satisfied: (1) There exists a label item $(u,w_s^1,w_e^1,k_1)\in L(v)$ which satisfies  $[w_s^1,w_e^1] \subseteq [w_s,w_e], k_1\le k$ . (2) There exist label items $(x,w_s^1,w_e^1,k_1)\\\in L(u)$, $(x,w_s^2,w_e^2,k_2)\in L(v)$ which satisfy  $\left[w_s^1,\ w_e^1\right]\subseteq\left[w_s,w_e\right]$, $[w_s^2,w_e^2] \subseteq [w_s,w_e]$, $k_1+k_2\le k$. For case (1), according to the pruning rule, it is known that the post-processed vertices can not be inserted into the path labels of the first-processed vertices as hop-points. For case (2) , it does not hold according to Theorem 2. Therefore the WKRI algorithm constructs an index that does not contain redundant path information.
\end{proof}
In Algorithm 1, assuming that the graph \textit{G} contains $|V|$ vertices and $|E|$ edges, the time complexity of completing the BFS traversal of all vertices is $O(|V|(|V|+|E|))$. Assuming that the upper bound of the number of label items associated with each vertex is \textit{l}, then during the traversal process, for each edge visited, the time complexity of determining whether the label item of a path containing the edge is redundant is $O(l)$, and hence the time complexity of Algorithm 1 is $O(l|V|(|V|+|E|))$. Assuming that each vertex stores at most \textit{L} labeled items, the space complexity of Algorithm 1 is $O(L|V|)$.

\subsection{Query Processing}
The query algorithm is processed based on the index labels, as shown in Algorithm 2. Given a query $Q(u,v,c,k)$, where $c=(w_s,w_e)$ is a weight constraint corresponding to three constraint forms. 1)It is corresponding to the constraint form of $\le w_e$ when $w_s$ is equal to $-\infty$. 2) It is corresponding to the constraint form of $\geq w_s$ when $w_e$ is equal to $+\infty$. 3)It is corresponding to the constraint form of $[w_s,w_e]$ otherwise. The specific process of querying can be described as follows: iteratively find whether there exists the same vertex in the index labels of \textit{u} and \textit{v}. If an identical point is found, the set of labels of the point will be examined to determine whether there exists a label item that satisfies the query condition. If it exists, the vertices \textit{u} and \textit{v} are reachable, otherwise the search for the next identical point continues.
\begin{algorithm}
\caption{Query processing algorithm Query()}
\begin{algorithmic}[1]
\renewcommand{\algorithmicrequire}{\textbf{Input:}}
\renewcommand{\algorithmicensure}{\textbf{Output:}}
\Require $Q(u,v,w_s,w_e,k)$ 
\Ensure true/false
\For{each $I_1=(u_1,w_s^1,w_e^1,k_1)\in L(u)$}
        \For{each $I_2=(u_2,w_s^2,w_e^2,k_2)\in L(v)$}
                \If{$u_1=u_2$ and $k_1+k_2\le k$}
                        \If{$[min(w_s^1,w_s^2),max(w_e^1,w_e^2)] \subseteq [w_s,w_e]$} 
                        \State \Return true
                        \EndIf
                \EndIf
        \EndFor{}
\EndFor{}
\end{algorithmic}
\end{algorithm}

\begin{example}
    Consider the indexes in Table \ref{tab3}. For the query (\textit{v}\textsubscript{2}, \textit{v}\textsubscript{6}, 5, 8, 3), we first find the same vertices \textit{v}\textsubscript{3} and \textit{v}\textsubscript{4} in the table by the index labels of \textit{v}\textsubscript{2} and \textit{v}\textsubscript{6}; then we proceed to detect whether there exists a label item in the table that satisfies the condition. Obviously, the label item of \textit{v}\textsubscript{2} (\textit{v}\textsubscript{3}, 4, 4, 1) and the label item of \textit{v}\textsubscript{6} (\textit{v}\textsubscript{3}, 8, 8, 1) can form the path $\textit{v}\textsubscript{2} \rightarrow \textit{v}\textsubscript{3} \rightarrow \textit{v}\textsubscript{6}$. The length of this path is 2, which satisfies the distance constraint (2 < 3). However, the range of the weights of this path is [4, 8], which does not satisfy the weight constraint [5, 8]. Therefore, we continue to checking the relevant label item of the next vertex \textit{v}\textsubscript{4}. Since there are (\textit{v}\textsubscript{4}, 5, 5, 1) and (\textit{v}\textsubscript{4}, 7, 7, 1) in the label items of \textit{v}\textsubscript{4}, and the paths formed by them satisfy the weight and distance constraints of the query, it can be known that \textit{v}\textsubscript{2} and \textit{v}\textsubscript{6} are reachable. Next, consider the query (\textit{v}\textsubscript{2}, \textit{v}\textsubscript{7}, 4, 5, 1), the same hop points are still \textit{v}\textsubscript{3} and \textit{v}\textsubscript{4}. However, none of the paths between \textit{v}\textsubscript{2} and \textit{v}\textsubscript{7} with \textit{v}\textsubscript{3} or \textit{v}\textsubscript{4} as intermediate points satisfy the weight and distance constraints of this query, and the query result is unreachable. Finally, consider the query (\textit{v}\textsubscript{6}, \textit{v}\textsubscript{7}, 6, 8, 2), based on the indexes in Table \ref{tab3}, it is known that the same hop points are \textit{v}\textsubscript{3}, \textit{v}\textsubscript{4}. Similar to the previous processing, Algorithm 2 first checks the related label items of \textit{v}\textsubscript{3}. Since none of the paths corresponding to the label items of \textit{v}\textsubscript{3} satisfy the weight and distance constraints of this query, it is necessary to continue to check the related label items of \textit{v}\textsubscript{4}. Since the paths consisting of label items (\textit{v}\textsubscript{4}, 7, 7, 1) and (\textit{v}\textsubscript{4}, 3, 3, 1) satisfy the distance constraints of the query, but the weight range is [3, 7], which does not satisfy the weight constraints [6, 8], the final query result is unreachable.
\end{example}
The query algorithm mainly detects whether two query points satisfy the constraints of the query by traversing the label sets of both query points. In the worst case, the label sets of both query points have to be traversed to give the query result. Since the first value of the indexed label item represents the hop point using the vertex processing order, and therefore the label items in the index are in ascending order. So the time complexity of Algorithm 2 is $O(L)$, where \textit{\textbf{L}} is the maximum value of the number of labeled items contained in all vertex index labels.

\section{ Optimization}
While the WKRI algorithm alleviates some of the limitations inherent in the basic algorithm, its requirement to traverse all vertices and construct an index during this traversal results in significant index construction time and an undesirably large index size. To address these challenges, we propose an index construction strategy inspired by the concept of minimum point coverage, as introduced by \citet{28brevsar2011minimum}. Although previous studies addressed a different problem domain, the underlying principle of minimizing coverage is leveraged here to effectively reduce the size of the index, thereby optimizing both the efficiency and scalability of the index construction process in our work.

\begin{definition}
(Point Coverage Set) For any edge $(u,v)\in E$, a vertex set $V^\ast$ is said to be a point coverage set of G if $\{u,v\}\cap V^\ast\neq\emptyset$. Among all the point covers of G, the one with the least number of vertices is called the minimal point cover.
\end{definition}
The advantages of constructing an index based on a minimum point coverage set are manifold. First, constructing an index based on a point coverage set is essentially a compression operation on the index. The reason is that the point coverage set is generally smaller than the vertex set \textit{V\textbf{ }}of the graph \textit{G}. Therefore, when constructing an index based on the point coverage set, the number of labeled items in the vertex labels can be reduced and thus reducing the index size. The smaller the size of the point coverage set, the smaller the size of the constructed index. Secondly, if a vertex does not belong to a point coverage set, all its neighbor vertices should belong to a point coverage set according to Definition 5. Therefore the shortest path information between any vertices in the graph will not be lost when constructing a index based on the vertices in the smallest point coverage set.

Sections 4.1 and 4.2 below discuss the two indexes GWKRI (Global WKRI) and LWKRI (Local WKRI) based on point coverage set respectively. Among them, GWKRI is a global index constructed by using the vertices in the minimum point coverage set as hop vertices, which has a smaller index size compared with WKRI indexe due to the compression of the number of vertices, and the query processing method is the same as that in Section 3.3. LWKRI also uses the vertices in the minimum point coverage set as hop vertices but only constructs indexes covering the vertices in the set, and thus has a smaller index size, but the querying process is relatively more complicated. It is a method of trading time for space.

\subsection{GWKRI Index Construction}
The GWKRI index construction process can be divided into two steps: (1) solving for minimum point coverage  (2) constructing the index based on the vertices in the minimum point coverage.

Since solving the minimum point coverage is an NP-hard problem\cite{28brevsar2011minimum}, in order to solve it quickly, we take a step back and solve the approximate minimum point coverage. We apply the following method to obtain an approximate solution in linear time.

Given a point coverage set \textit{M} with an initial value equal to the empty set and a sequence of vertices arranged in descending order of degree, the basic process of solving the approximate minimum point cover is to select the vertex \textit{u} with the largest degree and add it to \textit{M}, then delete \textit{u} and its associated edges and update the degrees of vertices affected by \textit{u} being deleted. This process is repeated until there are no edges in the graph, and the approximate minimum point coverage \textit{M} can be obtained.

Repeatedly selecting the vertex with the largest degree needs to maintain the order of vertex degrees, and updating the degree of the affected vertex will change the original ordering of vertex degrees. If you simply re-sort the vertices using the sorting algorithm, it will take more time. In order to solve this problem, this paper proposes an approximate minimum point coverage solution method with linear time complexity. The core idea of this method is to record the position of the last vertex among the vertices corresponding to each degree value i in the sorted array, so that when updating the degree of any vertex, the adjustment of the vertex position can be completed in $O(1)$ time. This ensures the correctness of sorting the vertices in the array in descending order of degree.

The process of solving the approximate minimum vertex cover is shown in Algorithm 3, in which 4 arrays are used, namely: (1) \textit{\textbf{node}} array, the size of which is also equal to $|V|$, which stores vertices sorted in descending order of degree, (2) ) \textit{\textbf{reversal}} array, the size of the array is equal to $|V|
$, the subscript corresponds to the vertex number, the element value represents the subscript of the vertex in the \textit{\textbf{node}} array, this array is used to save the mapping relationship between the vertex number and the subscript position, (3) \textit{\textbf{degree}} array, the size of this array is equal to $|V|$, which stores the degrees of all vertices, (4) \textit{\textbf{end}} array, the size of this array is equal to the maximum value of the degrees of all vertices in the graph plus 1. The element value with subscript \textit{i} is equal to the subscript of the last vertex with degree \textit{i} in the \textit{\textbf{node}} array. Lines 1-12 of Algorithm 3 initialize relevant variables.Lines 12-28 loop to solve the coverage set \textit{M}. The condition of the loop is to determine whether there is an edge in the graph. If it exists, lines 14-15 will find the point with the largest degree and add it to the coverage set \textit{M}. Lines 16-25 loop to process all the neighbors of \textit{u}. The basic idea is to reduce the degree of the neighbor by one, which is equivalent to the operation of deleting \textit{u} and its associated edges, and then place the neighbor with the degree reduced by one in the appropriate place in the \textit{\textbf{node}} array, so that the vertices in the \textit{\textbf{node}} array are still ordered in descending order of vertex degree.

\begin{algorithm}
\caption{Minimum Vertex Coverage Set Construction}
\begin{algorithmic}[1]
\renewcommand{\algorithmicrequire}{\textbf{Input:}}
\renewcommand{\algorithmicensure}{\textbf{Output:}}
\Require $\textit{G} 
=\left(V,E,\mathrm{\Sigma},w\right)$  
\Ensure {The minimum vertex coverage set \textit{M}}
\State $\textit{M}  \leftarrow  \emptyset$
\For{each $u\in V$}
        \State $end[degree[\textit{u}]] = end[degree[\textit{u}]] + 1$ 
\EndFor{}
\State $end[d_{max}] = end[d_{max}] - 1$
\For{$d\leftarrow d_{max} - 1\ to\  0$}
        \State $end[\textit{d}] = end[\textit{d}] + end[\textit{d}+1]$
\EndFor{}
\For{$i \leftarrow 0\ \ to\ \  |V|$}
        \State $reversal[node[\textit{i}]] = \textit{i} $
\EndFor{}
\State $e \leftarrow |E|, i \leftarrow 0$
\While{$e > 0$}
        \State $\textit{u} = node[\textit{i}]$ 
        \State $M = M \cup u$
        \For{each $v \in N(u)$}
                \State $\textit{d} \leftarrow degree[\textit{v}]$
                \State $\textit{temp} \leftarrow node[end[\textit{d}] \ -\  1]$ 
                \State $node[end[\textit{d}] \ -\  1] \leftarrow \textit{v} $
                \State $node[reversal[\textit{v}]] \leftarrow \textit{temp} $
                \State $reversal[\textit{temp}] \leftarrow reversal[\textit{v}] $
                \State $reversal[\textit{v}] \leftarrow end[\textit{d}] \ -\  1 $
                \State $end[\textit{d}] \leftarrow end[\textit{d}] \ -\  1 $
                \State $degree[\textit{v}]\leftarrow\textit{d }\ -\  1 $
        \EndFor{}
        \State $\textit{e} = \textit{e} \ -\  degree[\textit{u}] $
        \State $\textit{i} = \textit{i} + 1$
\EndWhile{}
\State\Return \textit{M}
\end{algorithmic}
\end{algorithm}
Figure 3 shows the specific process of how index array \textit{\textbf{end}} maintaining the sorted array. Observing Figure 3(a), we can see that vertex $v_3$ is the point with the largest degree, so vertex $v_3$ is first added to the coverage set, and the degrees of its neighbor vertices $v_1,v_2,v_5,v_6$ are updated at the same time. For vertex $v_1$, its degree before updating is 3, so exchange $v_1$ with vertex $v_4$ whose subscript is equal to $\textit{end}[3]$ in the \textit{\textbf{node}} array. After that, the degree of $v_1$ is reduced by one, and the value of the element with subscript  3 in \textbf{\textit{end}} is reduced by one, which means that there is one less vertex with degree  3. The next vertex to be processed is\textit{ b}. Since its degree is  3, and the element with subscript 3 in the \textit{\textbf{end}} array is  2, corresponding to the vertex with subscript  2 in the \textbf{\textit{node}} array, it can be found that it is $v_2$ itself, so there is no need to exchange. Then reduce the degree of $v_2$ by one, and further reduce the value of the element with index 3 in the \textbf{\textit{end}} array by one. Vertices $v_5$ and $v_6$ are then processed. The specific process is similar. The state after processing these four vertices is shown in Figure 3(b).
\begin{figure}[htbp]
    \centering
    \subfloat[Before $v_3$ is added to the coverage set]
    {
    \includegraphics[width=0.5\linewidth]{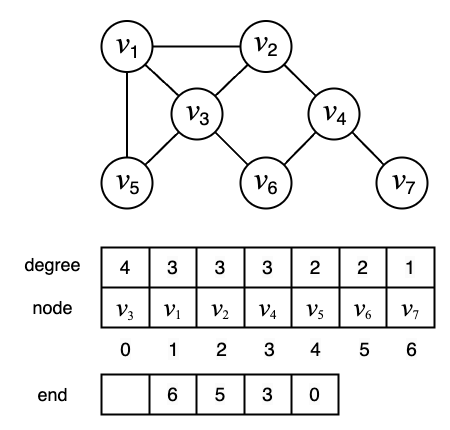}
    }
    \subfloat[]
    {
    \includegraphics[width=0.5\linewidth]{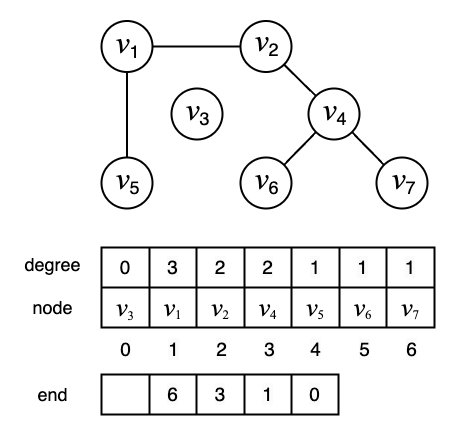}
    }
    \caption{Linear update method for vertex degree ordering of minimum coverage set}
\end{figure}

Algorithm 3 uses four one-dimensional arrays to cache intermediate results during operation, including \textit{\textbf{end}}, \textbf{\textit{node}}, \textbf{\textit{degree}} and \textit{\textbf{reverse}} arrays. The upper bound on the size of these four arrays is the number of vertices of the input data graph, so the space complexity of Algorithm 3 is $O(|V|)$. The cost of the algorithm in the initialization phase is related to the number of vertices in the graph, so the cost of the initialization phase is $O(|V|)$. In the coverage set solving phase, each edge needs to be visited once, and the cost of processing a single edge at the same time is $O(1)$, so the cost of this phase is $O(|E|)$. In summary, the time complexity of Algorithm 3 is $O(|V|+|E|)$.

For the vertices in the minimum point cover, the vertex processing order is obtained by arranging them in descending order according to their degree. Select vertices as hop points according to the processing order for BFS traversal, and the index labels of all vertices can be established. The construction process is consistent with the WKRI algorithm. Example 4 details the process of constructing the GWKRI index using minimum point coverage.
\begin{example}
    Considering the undirected weighted graph $G$ in Figure 1, the approximate minimum point coverage $M$ of the vertices in the graph is first obtained through Algorithm 3, and the vertex processing order \textit{v}\textsubscript{3}, \textit{v}\textsubscript{4}, \textit{v}\textsubscript{1} in $M$ is obtained according to the degree value in descending order. First, \textit{v}\textsubscript{3} is processed. The specific method is to perform the traversal operation starting from \textit{v}\textsubscript{3}. For each vertex encountered during the traversal process, insert the label item $I\ =\ (\textit{v}\textsubscript{3},w_s,w_e,k)$ into the label index of the vertex. It should be noted that the pruning rules introduced in the previous chapter are also used during the traversal process to avoid redundant operations. After processing \textit{v}\textsubscript{3}, continue to process other vertices until all vertices in the coverage set have been processed. The constructed index is shown in Table \ref{tab4}. Obviously, the number of vertices in the coverage set has a greater impact on the index size. Compared with the WKRI algorithm that needs to process all vertices to build the index, the GWKRI index construction process only needs to process the vertices in the coverage set, and the index construction time and scale are significantly reduced.
\end{example}
\begin{table}[htbp]
  \centering
  \caption{GWKRI index}
    \begin{tabular}{cccc}
    \toprule
    Label & \textit{v}\textsubscript{3} & \textit{v}\textsubscript{4} & \multicolumn{1}{c}{\textit{v}\textsubscript{1}} \\
    \midrule
    \multirow{3}[2]{*}{\textit{L}(\textit{v}\textsubscript{1})} & (\textit{v}\textsubscript{3},3,3,1) & \multirow{3}[2]{*}{(\textit{v}\textsubscript{4},5,5,2)} & \multicolumn{1}{c}{\multirow{3}[2]{*}{(\textit{v}\textsubscript{1},0,0,0)}} \\
    \multicolumn{1}{c}{} & (\textit{v}\textsubscript{3},4,5,2) & \multicolumn{1}{c}{} &  \\
    \multicolumn{1}{c}{} & (\textit{v}\textsubscript{3},5,8,4) & \multicolumn{1}{c}{} &  \\
    \midrule
    \multirow{2}[2]{*}{\textit{L}(\textit{v}\textsubscript{2})} & (\textit{v}\textsubscript{3},4,4,1) & \multirow{2}[2]{*}{(\textit{v}\textsubscript{4},5,5,1)} & \multicolumn{1}{c}{\multirow{2}[2]{*}{(\textit{v}\textsubscript{1},5,5,1)}} \\
    \multicolumn{1}{c}{} & (\textit{v}\textsubscript{3},5,8,3) & \multicolumn{1}{c}{} &  \\
    \midrule
    \textit{L}(\textit{v}\textsubscript{3})  & (\textit{v}\textsubscript{3},0,0,0) & \multicolumn{1}{c}{} &  \\
    \midrule
    \multirow{2}[2]{*}{\textit{L}(\textit{v}\textsubscript{4})} & (\textit{v}\textsubscript{3},4,5,2) & \multirow{2}[2]{*}{(\textit{v}\textsubscript{4},0,0,0)} & \multirow{2}[2]{*}{} \\
    \multicolumn{1}{c}{} & (\textit{v}\textsubscript{3},7,8,2) & \multicolumn{1}{c}{} &  \\
    \midrule
    \multirow{2}[2]{*}{\textit{L}(\textit{v}\textsubscript{5})} & (\textit{v}\textsubscript{3},6,6,1) & \multirow{2}[2]{*}{(\textit{v}\textsubscript{4},2,5,3)} & \multicolumn{1}{c}{\multirow{2}[2]{*}{(\textit{v}\textsubscript{1},2,2,1)}} \\
    \multicolumn{1}{c}{} & (\textit{v}\textsubscript{3},2,3,2) & \multicolumn{1}{c}{} &  \\
    \midrule
    \multirow{2}[2]{*}{\textit{L}(\textit{v}\textsubscript{6})} & (\textit{v}\textsubscript{3},8,8,1) & \multirow{2}[2]{*}{(\textit{v}\textsubscript{4},7,7,1)} & \multirow{2}[2]{*}{} \\
    \multicolumn{1}{c}{} & (\textit{v}\textsubscript{3},4,7,3) & \multicolumn{1}{c}{} &  \\
    \midrule
    \textit{L}(\textit{v}\textsubscript{7})  & (\textit{v}\textsubscript{3},3,5,3) & (\textit{v}\textsubscript{4},3,3,1) &  \\
    \bottomrule
    \end{tabular}%
  \label{tab4}%
\end{table}%

\subsection{LWKRI Index Construction }
In order to further reduce the index size, we propose the LWKRI index. Compared with GWKRI, the LWKRI index only builds index labels for the vertices in the coverage set. Therefore, based on GWKRI, it is necessary to determine whether the currently visited vertex belongs to the vertex in the coverage set. If it does, determine whether the path is redundant and decide whether to insert the index,  otherwise directly record the relevant information of the path and enter the next round of BFS traversal, as shown in Algorithm 4.

\begin{algorithm}
\caption{LWKRI Index Construction Algorithm}
\begin{algorithmic}[1]
\renewcommand{\algorithmicrequire}{\textbf{Input:}}
\renewcommand{\algorithmicensure}{\textbf{Output:}}
\Require $\textit{G} 
=\left(V,E,\mathrm{\Sigma},w\right)$ , \textit{\textbf{M}}
\Ensure Index label \textit{L}(\textit{u}) for each vertex \textit{u}
\For{each $\textit{v}\in \textit{M}$(from higher order to lower)}
        \State Perform BFS starting from \textit{v}, and for each  vertex \textit{u} being visited: 
        \State $I\gets(v,w_s,w_e,k)$
        \If{$v\ \in\ M$}
                \If{isRedundant\textbf{(}\textit{I}\textbf{) = }false}
                        \State add \textit{I }to \textit{L}(\textit{u}) and expand \textit{u} in the BFS 
                \Else{ Do not add \textit{I} to \textit{L}(\textit{u}) nor expand \textit{u} }
                \EndIf
        \Else{}
                \textbf{continue}
        \EndIf
\EndFor{}
\end{algorithmic}
\end{algorithm}

\begin{example}
    Considering the undirected weighted graph \( G \) in Figure 1, the vertex processing order \textit{v}\textsubscript{3}, \textit{v}\textsubscript{4}, \textit{v}\textsubscript{1} in \( M \) is obtained by arranging it in descending order according to the degree value. When processing \textit{v}\textsubscript{3}, four vertices \textit{v}\textsubscript{1}, \textit{v}\textsubscript{2}, \textit{v}\textsubscript{5}, \textit{v}\textsubscript{6} need to be accessed. When \textit{v}\textsubscript{1} is encountered in the process of processing \textit{v}\textsubscript{3}, it will be judged whether the resulting path is redundant. If it is redundant, it will be discarded. Otherwise, the label item about \textit{v}\textsubscript{3} will be inserted into the index label of \textit{v}\textsubscript{1}. Since \textit{v}\textsubscript{2}, \textit{v}\textsubscript{5}, and \textit{v}\textsubscript{6} are not in the coverage set, no additional judgment is required when accessing \textit{v}\textsubscript{2}, \textit{v}\textsubscript{5}, and \textit{v}\textsubscript{6}. You can directly record the path information related to the path, and then continue traversing in the same way. Other vertices are processed similarly, and the final LWKRI index is shown in Table \ref{tab5}.
\end{example}

\begin{table}[htbp]
  \centering
  \caption{LWKRI index }
    \begin{tabular}{cccc}
    \toprule
    \multicolumn{1}{c}{Label} & \textit{v}\textsubscript{3} & \multicolumn{1}{c}{\textit{v}\textsubscript{4}} & \multicolumn{1}{c}{\textit{v}\textsubscript{1}} \\
    \midrule
    \multicolumn{1}{c}{\multirow{3}[2]{*}{\textit{L}(\textit{v}\textsubscript{1})}} & (\textit{v}\textsubscript{3},3,3,1) & \multicolumn{1}{c}{\multirow{3}[2]{*}{(\textit{v}\textsubscript{4},5,5,2)}} & \multicolumn{1}{c}{\multirow{3}[2]{*}{(\textit{v}\textsubscript{1},0,0,0)}} \\
          & (\textit{v}\textsubscript{3},4,5,2) &       &  \\
          & (\textit{v}\textsubscript{3},5,8,4) &       &  \\
    \midrule
    \multicolumn{1}{c}{\textit{L}(\textit{v}\textsubscript{3})} & (\textit{v}\textsubscript{3},0,0,0) &       &  \\
    \midrule
    \multicolumn{1}{c}{\multirow{2}[2]{*}{\textit{L}(\textit{v}\textsubscript{4})}} & (\textit{v}\textsubscript{3},4,5,2) & \multicolumn{1}{c}{\multirow{2}[2]{*}{(\textit{v}\textsubscript{4},0,0,0)}} & \multirow{2}[2]{*}{} \\
          & (\textit{v}\textsubscript{3},7,8,2) &       &  \\
    \bottomrule
    \end{tabular}%
  \label{tab5}%
\end{table}%

Table \ref{tab5} gives the LWKRI index. Compared with the GWKRI index of Table \ref{tab4}, the size of the LWKRI index is significantly reduced. Assume that any hop vertex generates at most \textit{l} label items for each other vertex. Since the LWKRI index only indexes vertices in the coverage set, the upper bound on the number of label items for each vertex is $|M|\times l$, so space complexity is $O(l{|M|}^2)$. The time complexity of Algorithm 4 for building the LWKRI index is consistent with the complexity of building the GWKRI index. The difference is that the LWKRI index only stores label items related to the vertices in the minimum point coverage set \textit{M}.

\subsection{Query Processing}
The query processing method based on the GWKRI index is the same as the query processing method based on the WKRI index. The query processing algorithm in this section mainly introduces the LWKRI index, which is divided into the following three situations: (1) Both vertices in the query belong to the minimum point coverage; (2) Only one of the two vertices belongs to the minimum point cover; (3) Neither of the two vertices belongs to the minimum point cover.

In order to handle these three situations uniformly, we take advantage of the following properties of the minimum point cover set \textit{M}: if the vertex $u\notin M$, then all neighbors of \textit{u} belong to \textit{M}. This property can be simply proved by contradiction: if there is a neighbor $v\notin M$ of u, then ${u,v}\cap M=\emptyset$, which violates the definition of a point coverage set.

Therefore, for the first case, it can be answered directly through the index, as shown in line 1-2 of Algorithm 5. The Query() function in the algorithm corresponds to the query algorithm in Algorithm 2. For the second case, as shown in lines 3-16 of Algorithm 5, the query can be transformed into a \textit{k}-1 step reachable query for neighbor vertices of non-minimum point covering set vertices. The third case corresponds to lines 17-25 of Algorithm 5. First, find the neighbor vertices of vertices \textit{u} and \textit{v} in the coverage set, then change \textit{k} in the query condition to \textit{k}-2, we can process the reachability between the neighbors of the two query points to answer the reachability between \textit{u} and \textit{v}. The specific process is shown in Algorithm 5.

\begin{algorithm}
\caption{Query processing algorithm of LWKRI}
\begin{algorithmic}[1]
\renewcommand{\algorithmicrequire}{\textbf{Input:}}
\renewcommand{\algorithmicensure}{\textbf{Output:}}
\Require $Q(u,v,w_s,w_e,k)$ 
\Ensure true/false

\If{$u\in M\land v\in M$}
\State \Return $Query(Q)$
\ElsIf{$u\in M\land v\notin M$}
        \For{each $x\in N(v)$}
                \State$w_x \leftarrow w(e_{x\rightarrow v})$
                \If{$Query\left(u,x,w_s,w_e,k-1\right)\ \  \land  \ \ w_x\in [w_s,w_e]$}
                \State \Return true
                \EndIf
        \EndFor{}
\ElsIf{$u\notin M\land v\in M$}
        \For{each $y\in N(u)$}
                \State$w_y \leftarrow w(e_{y\rightarrow u})$
                \If{$Query\left(y,v,w_s,w_e,k-1\right)\ \  \land  \ \ w_y\in [w_s,w_e]$}
                \State \Return true
                \EndIf
        \EndFor{}
\ElsIf{$u\notin M\land v\notin M$}
        \For{each $y\in N(u)$}
                \For{each $x\in N(v)$}
                        \State $w_x \leftarrow w(e_{x\rightarrow v})$; $w_y \leftarrow w(e_{y\rightarrow u})$
                        \If{$Query\left(y,x,w_s,w_e,k-2\right)\ \  \land  \ \  w_x,w_y\in [w_s,w_e]$}
                        \State \Return true
                        \EndIf
                \EndFor{}
        \EndFor{}
\EndIf
\State\Return false
\end{algorithmic}
\end{algorithm}

The worst-case time complexity of Algorithm 5 is mainly reflected in the third case, when answering a given query requires processing $d(u)\times(d(v)$ queries based on the neighbors of the query point.Since the index label of a vertex contains at most $|M|$ hop points, the cost of finding common hop points is $O(|M|)$. In addition, the cost of judging \textit{k} constraints is $O(1)$, and the cost of judging weight constraints is $O(\left|\mathrm{\Sigma}\right|)$, so the worst-case time complexity of Algorithm 5 is $O (d(u)\times d(v)\times |M| \times|\mathrm{\Sigma}|)$.

Although the algorithm proposed in this article studies undirected weighted graphs, it can be extended to directed weighted graphs. In terms of algorithm design, the difference between undirected graphs and directed graphs is as follows: in terms of constructing a 2-hop index, there are two labels for each vertex \textit{v}, one label contains the vertices that can be reached by \textit{v}, and the other label contains the vertices that can reach \textit{v}. An undirected graph contains only one index label. Therefore, when building an index, a directed graph only has one more round of reverse traversal than an undirected graph.

\section{Experiment}
\subsection{Experimental Setup}
The experiment in this article is run on a large server with 20 cores, Intel Xeon Gold 5218R CPU clocked at 2.1GHz, and 1024G DDR4 memory. The programming language is C++, the compilation environment is g++ 9.4.0, and the operating system is Ubuntu 20.04.

There are three indexes used for comparison in the experiment, including (1) WKRI index, (2) GWKRI index, and (3) LWKRI index. The main evaluation indicators of the experiment include index size, index construction time and query response time.

\subsection{Dataset}
The 15 real data sets used in the experiments of this article include Email-EuAll, web-Google, soc-LiveJournall, WikiTalk, com-Orkut, cit-Patents\footnote{http://snap.stanford.edu/}, uniprot22m, uniprot100m, uniprot150m, go-uniprot, 10go-uniprot\footnote{http://www.uniprot.org/}, dbpedia\footnote{http://dbpedia.org/}, twitter\footnote{http://twitter.mpi-sws.org/}, Amazon and Web trackers\footnote{http://konect.cc/}. Among them, com-Orkut and Web trackers contain more than 100 million edges. Table \ref{tab6} shows the relevant information of these data sets, where the size of the graph is $|V|+|E|$, and the number of vertices \textit{$|M|$} in the covering set is the number of vertices in the approximate minimum covering set calculated using the method introduced in Section 4.1.In addition, based on the above original dataset, this paper also changes the range of the weight set from $|\mathrm{\Sigma}|=10$ to $|\mathrm{\Sigma}|=100$, $|\mathrm{\Sigma}|=200$, and $|\mathrm{\Sigma}|=500$ while keeping the number of vertices and edges unchanged, and it is used to test the effect of the change of the range of weights on the performance of this paper's algorithm.

\begin{table}[htbp]
  \centering
  \caption{Dataset statistics  ($|\mathrm{\Sigma}|=10$) }
    \begin{tabular}{cccccc}
    \toprule
    Dataset & \multicolumn{1}{c}{\textit{$|V|$}} & \multicolumn{1}{c}{\textit{$|E|$}} & \multicolumn{1}{c}{graph size} & \multicolumn{1}{c}{average degree} & \multicolumn{1}{c}{\textit{$|M|$}} \\
    \midrule
    Email-EuAll & 231,000 & 223,004 & 454,004 & 1.93  & 16,875 \\
    \midrule
    web-Google & 371,753 & 517,805 & 889,558 & 2.79  & 106,481 \\
    \midrule
    soc-LiveJournall & 971,232 & 1,024,140 & 1,995,372 & 2.12  & 85,476 \\
    \midrule
    uniprot22m & 1,595,444 & 1,595,442 & 3,190,886 & 2.01& 141,517 \\
    \midrule
    WikiTalk & 2,281,879 & 2,311,570 & 4,593,449 & 2.03  & 167,440 \\
    \midrule
    dbpedia & 3,365,623 & 7,989,191 & 11,354,814 & 4.75  & 403,874 \\
    \midrule
    uniprot100m & 16,087,295 & 16,087,293 & 32,174,588 & 2.01& 1,366,073 \\
    \midrule
    twitter & 18,121,168 & 18,359,487 & 36,480,655 & 2.03  & 423,057 \\
    \midrule
    uniprot150m & 25,037,600 & 25,037,598 & 50,075,198 & 2.01  & 2,148,371 \\
    \midrule
    10go-uniprot & 469,526 & 3,476,397 & 3,945,923 & 15.03 & 161,038 \\
    \midrule
    cit-Patents & 3,774,768 & 16,518,947 & 20,293,715 & 8.75  & 1,726,788 \\
    \midrule
    go-uniprot & 6,967,956 & 34,769,339 & 41,737,295 & 10.12 & 20,155 \\
    \midrule
    Amazon & 31,050,733 & 82,677,131 & 113,727,864 & 5.33& 4,316,214 \\
    \midrule
    Web trackers & 40,421,974 & 140,613,762 & 181,035,736 & 6.91& 5,101,631 \\
    \midrule
    com-Orkut & 3,072,441 & 117,185,083 & 120,257,524 & 18.42 & 463,214 \\
    \bottomrule
    \end{tabular}%
  \label{tab6}%
\end{table}%

\subsection{Performance Analysis and Comparison}
Since this paper is the first to propose a simultaneous solution to the weight-based and distance-constrained reachability (WCKR) query problem, the comparative experiments constructed focus on online search-based methods such as BFS traversal.
At the same time, a large number of experimental comparisons between the baseline method (WKRI) and the optimization methods (GWKRI, LWKRI) proposed in this paper are also carried out to analyze the index construction time, index size, and the corresponding query time of the different methods. The experimental results verify the high efficiency of the baseline method and the optimization method in this paper.
\begin{figure}[h]
    \centering
    \includegraphics[width=0.75\linewidth]{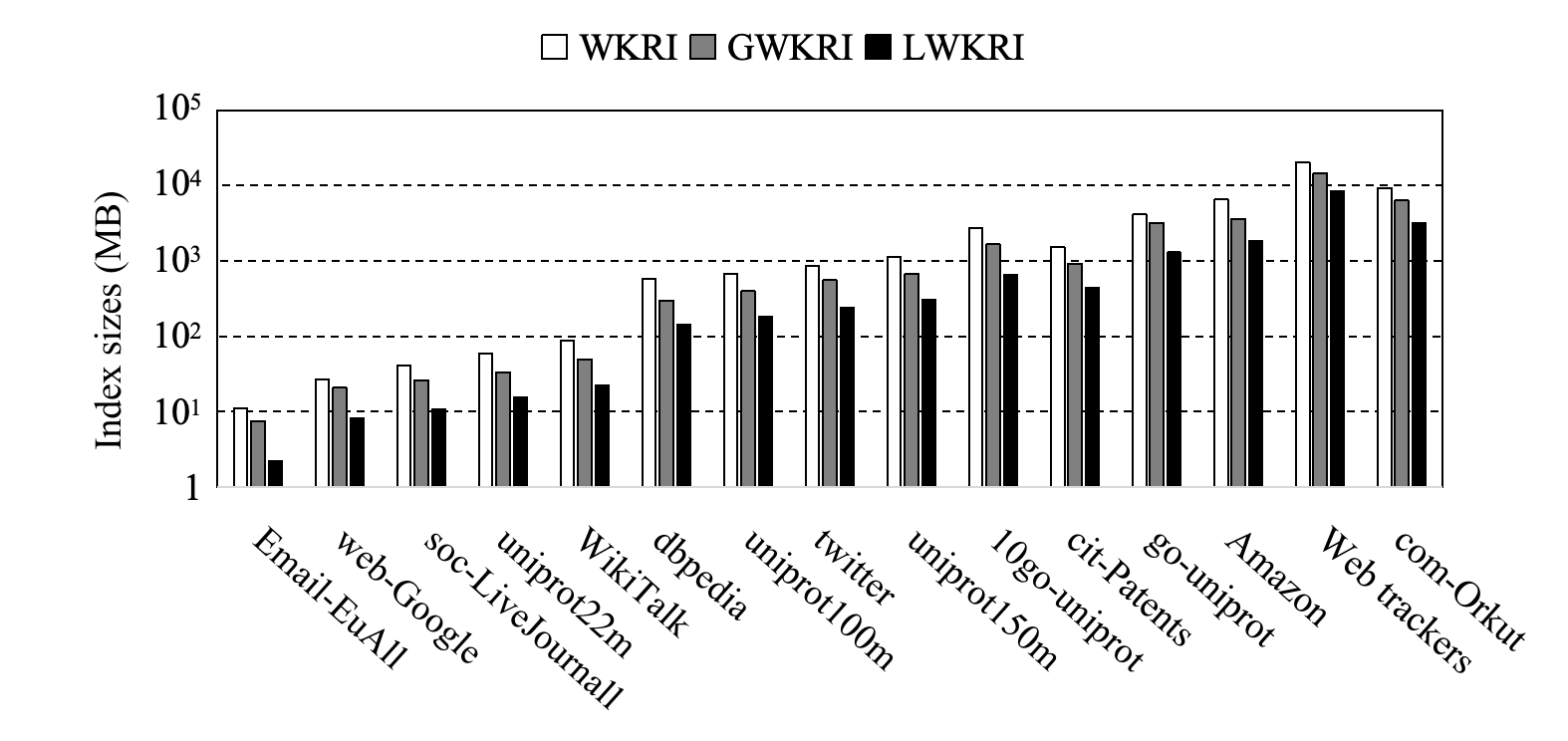}
    \caption{Index size}
    \label{fig:4}
\end{figure}
\subsubsection{Index Size}
Figure \ref{fig:4} demonstrates the comparison of the three indexes in terms of index size. As depicted in the figure, the sizes of the three indexes generally decrease across the datasets, with the LWKRI index consistently achieving the smallest size. This trend suggests that the indexing strategy employed by LWKRI is more efficient in minimizing storage requirements compared to WKRI and GWKRI.

The data shows that in some datasets, the LWKRI index is approximately 5 times smaller than the WKRI index and about 2 times smaller than the GWKRI index. This significant reduction in size indicates the effectiveness of the LWKRI strategy, which is based on minimum vertex coverage, in optimizing index storage.

The figure illustrates that the extent of the size reduction varies across different datasets. While the LWKRI index consistently outperforms the other two indexes, the degree of improvement differs, suggesting that the benefits of the minimum vertex coverage strategy may be more pronounced in certain types of datasets.

\begin{figure}[h]
    \centering
    \includegraphics[width=0.75\linewidth]{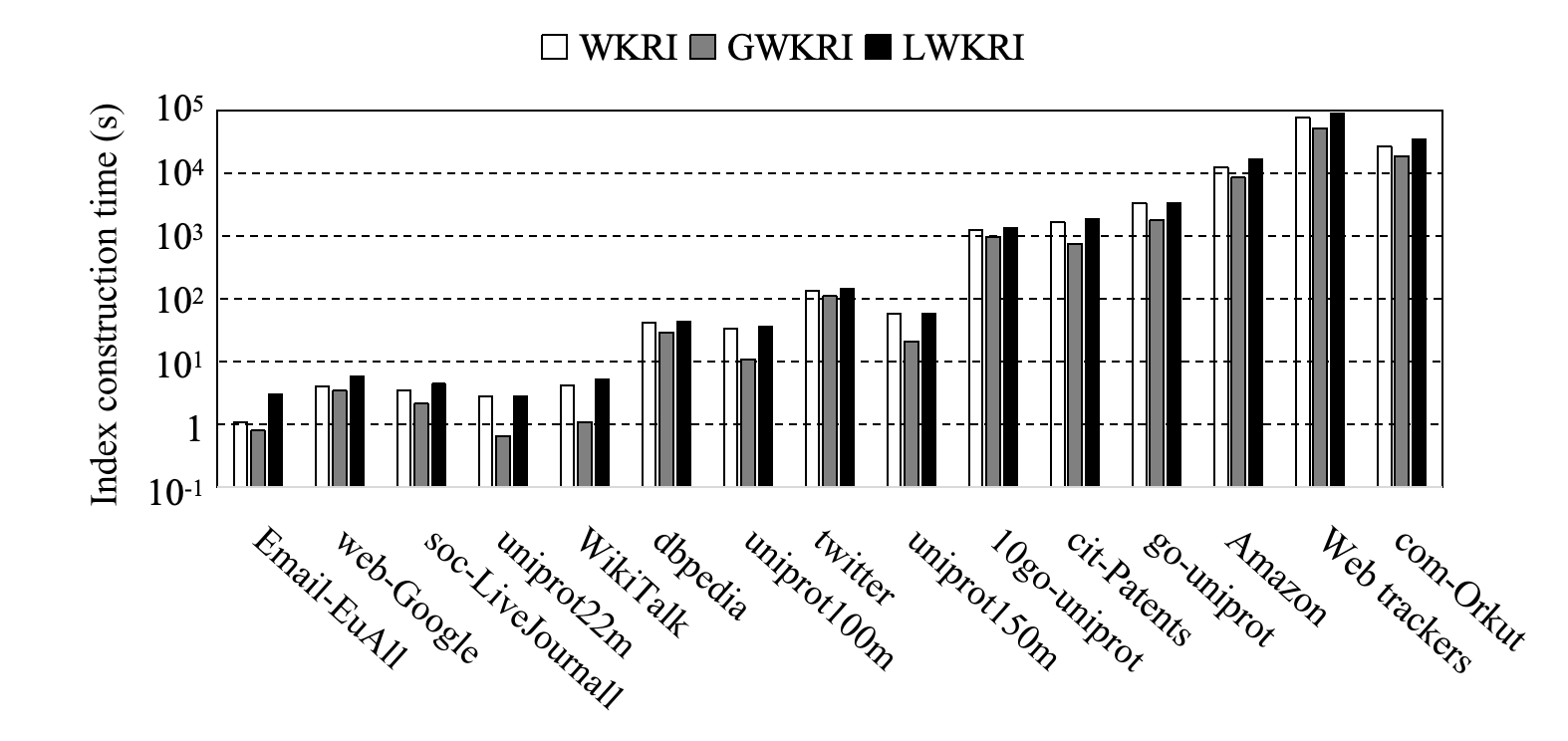}
    \caption{Index Construction time}
    \label{fig:5}
\end{figure}

\subsubsection{Index Construction Time}
In this section, we compare the construction times of three indexing schemes: GWKRI, LWKRI, and WKRI. The experimental results, illustrated in Figure \ref{fig:5}, reveal notable differences in performance. Specifically, the GWKRI index is approximately four times faster than the WKRI index on sparse graphs. This efficiency is attributed to the fact that the GWKRI index only requires traversal operations from vertices in the minimum vertex cover set, whereas the WKRI index necessitates traversal from every vertex in the graph.

As the graph density increases, the size of the minimum vertex cover set grows, which means that by the time vertices in this set are processed, many labeled items have already been identified. Consequently, the need for additional traversal from non-cover vertices diminishes, leading to a convergence in construction times between the WKRI and GWKRI indices.

The LWKRI index shows construction times similar to those of the WKRI index. Despite LWKRI storing index label entries solely for vertices in the minimum vertex cover, it suffers from increased traversal overhead. This is because LWKRI cannot effectively prune vertices not in the cover set during traversal, resulting in additional processing time. Thus, while LWKRI processes fewer vertices, the traversal inefficiencies contribute to a time overhead comparable to that of WKRI.

For the com-Orkut dataset, both GWKRI and WKRI indices exhibit longer construction times due to the high edge count and increased graph density. This observation highlights the impact of graph density on index construction efficiency and underscores the necessity for optimized indexing strategies in dense graph scenarios.

\begin{figure}[h]
    \centering
    \includegraphics[width=0.75\linewidth]{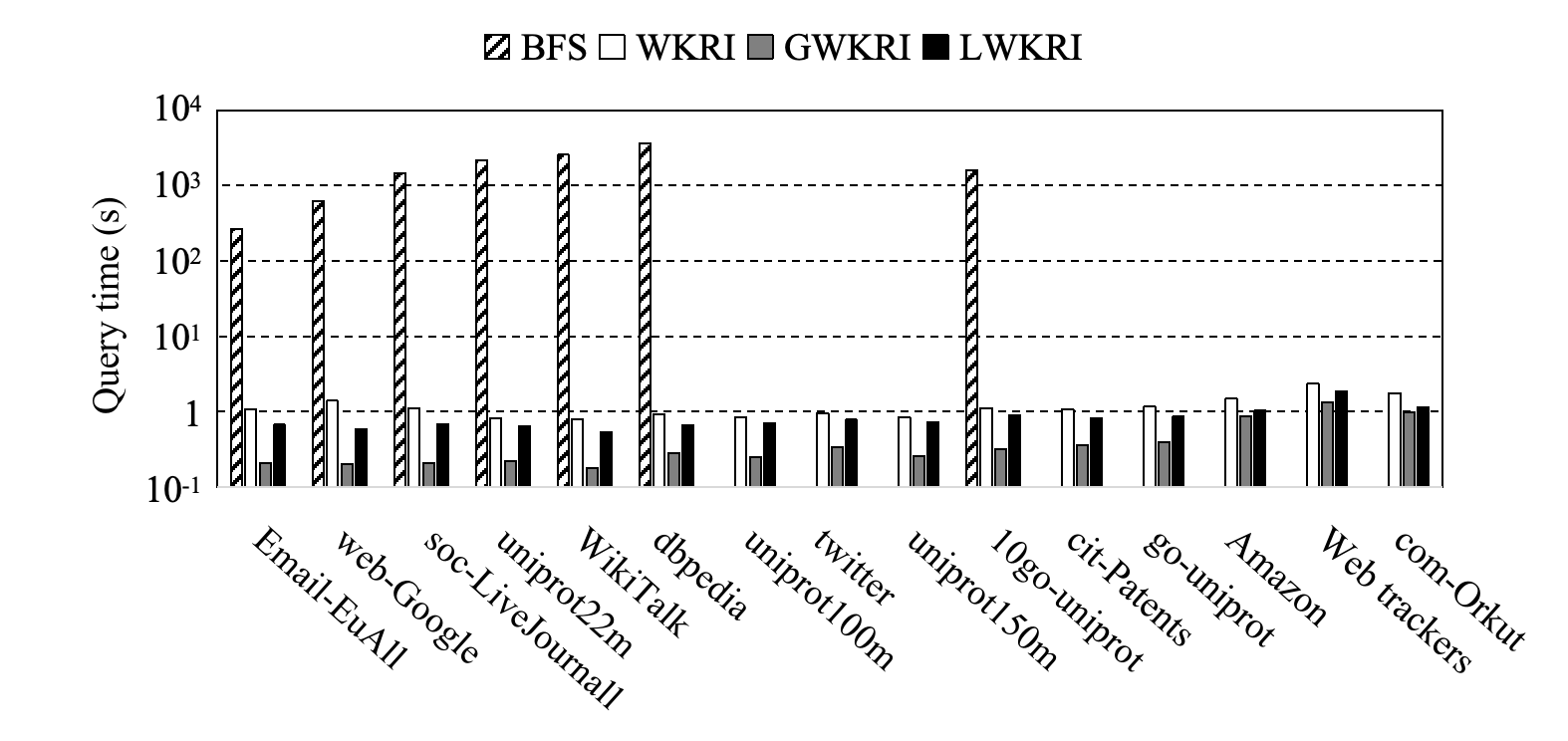}
    \caption{Query time}
    \label{fig:6}
\end{figure}

\subsubsection{Query Response Time}
For each dataset, a collection containing 1 million queries was randomly generated, comprising 500,000 reachable and 500,000 unreachable queries. Figure \ref{fig:6} presents a comparison of the query response times of the four algorithms.

As shown in the figure, the BFS algorithm is significantly less efficient than the other three algorithms in terms of query time. For large graphs, the BFS algorithm's query time can exceed 5000 seconds, which is indicated by "/" in Figure \ref{fig:6} and is therefore not discussed in detail.

Although the query processing methods for the WKRI and GWKRI indexes are similar, the query performance of the GWKRI index is approximately 6 times faster than that of the WKRI index. This difference in performance is primarily attributed to the larger index size of WKRI, which increases the traversal cost of index label items during query processing.

For the LWKRI index, its query efficiency is slightly lower than that of the GWKRI index. This is mainly because, with the GWKRI index, only one query needs to be processed. In contrast, when using the LWKRI index, if the query point is not part of the minimum vertex coverage, multiple queries must be constructed and processed based on the neighboring vertices of the query point. This additional step introduces a relatively large time overhead.

Overall, while the LWKRI index offers significant storage efficiency, this comes with a trade-off in query performance, particularly when compared to the GWKRI index.

\subsubsection{Comparison of Algorithm Efficiency under Different Weight Scales}
In this paper, we conduct comparative experiments on index construction efficiency by varying the range of $|\mathrm{\Sigma}|$. Specifically, there are four cases of $|\mathrm{\Sigma}|=10$, $|\mathrm{\Sigma}|=100$,$|\mathrm{\Sigma}|=200$, and $|\mathrm{\Sigma}|=500$, in which the dataset of $|\mathrm{\Sigma}|=10$ is the original dataset. The latter three cases are to reassign the weights on each edge randomly from $0-X (X=100, 200, 500)$ while keeping the vertices and edges unchanged. The experimental results are shown in Figs. 4-6. We use four representative datasets (soc-LiveJournall, uniprot150m, cit-Patents, and 10go-uniprot) to compare the effect of the range of $|\mathrm{\Sigma}|$ on the index construction efficiency. Soc-LiveJournall and uniprot150m are sparser and the index building efficiency is less affected by the range of $|\mathrm{\Sigma}|$, cit-Patents and 10go-uniprot are denser and the index building efficiency is more affected by the range of $|\mathrm{\Sigma}|$.

\begin{figure}[h]
  \centering
  {\includegraphics[width=0.5\textwidth]{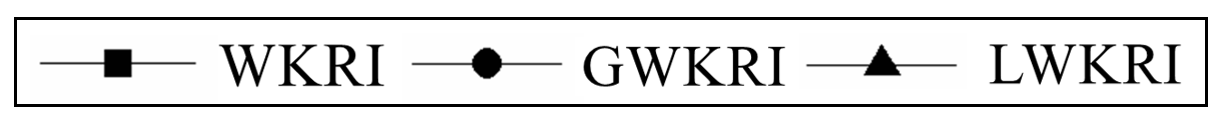}}
  \centering
  \subfloat[soc-LiveJournall]
  {\includegraphics[width=0.25\textwidth]{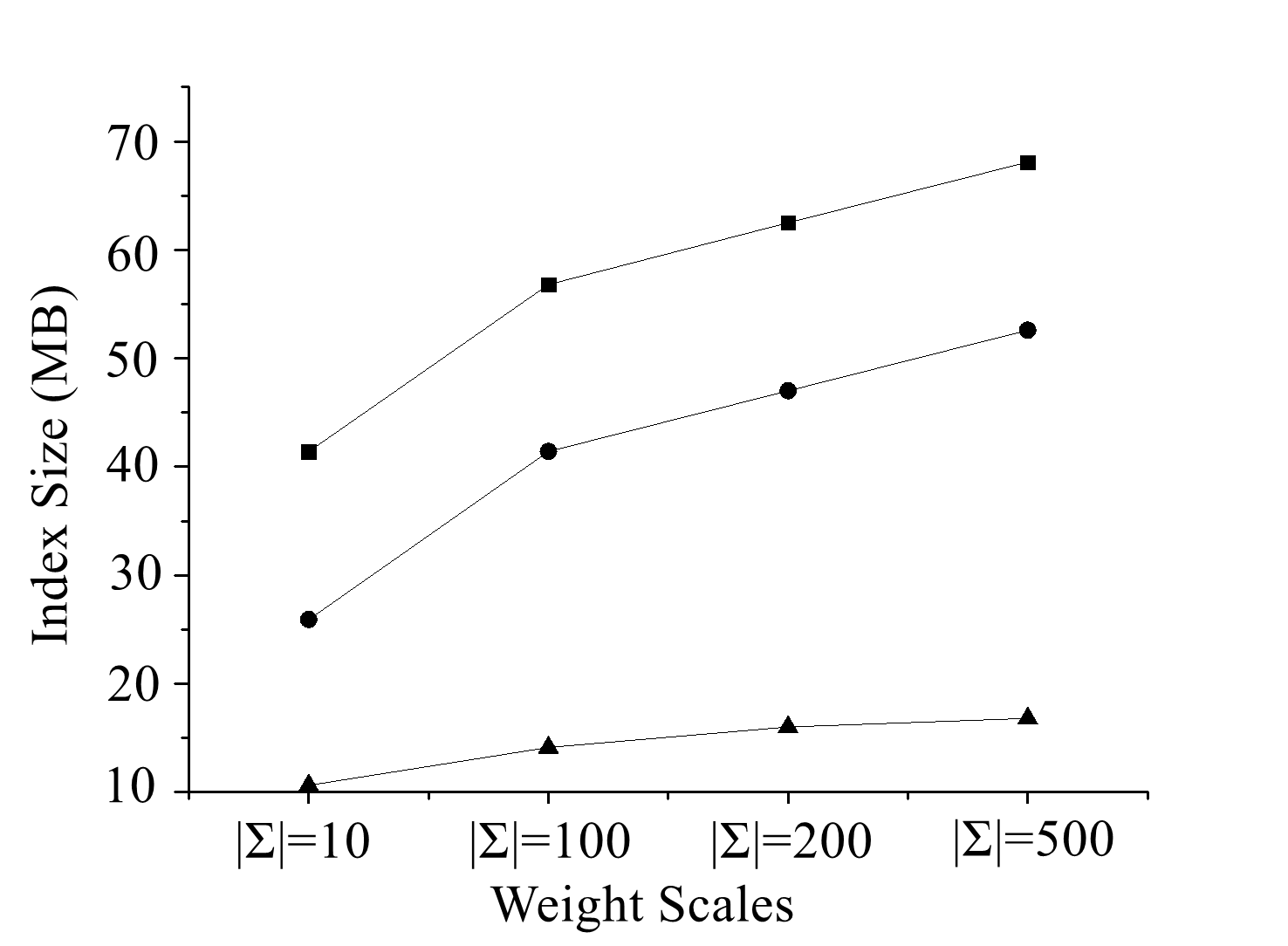}}    
  \subfloat[uniprot150m]
  {\includegraphics[width=0.25\textwidth]{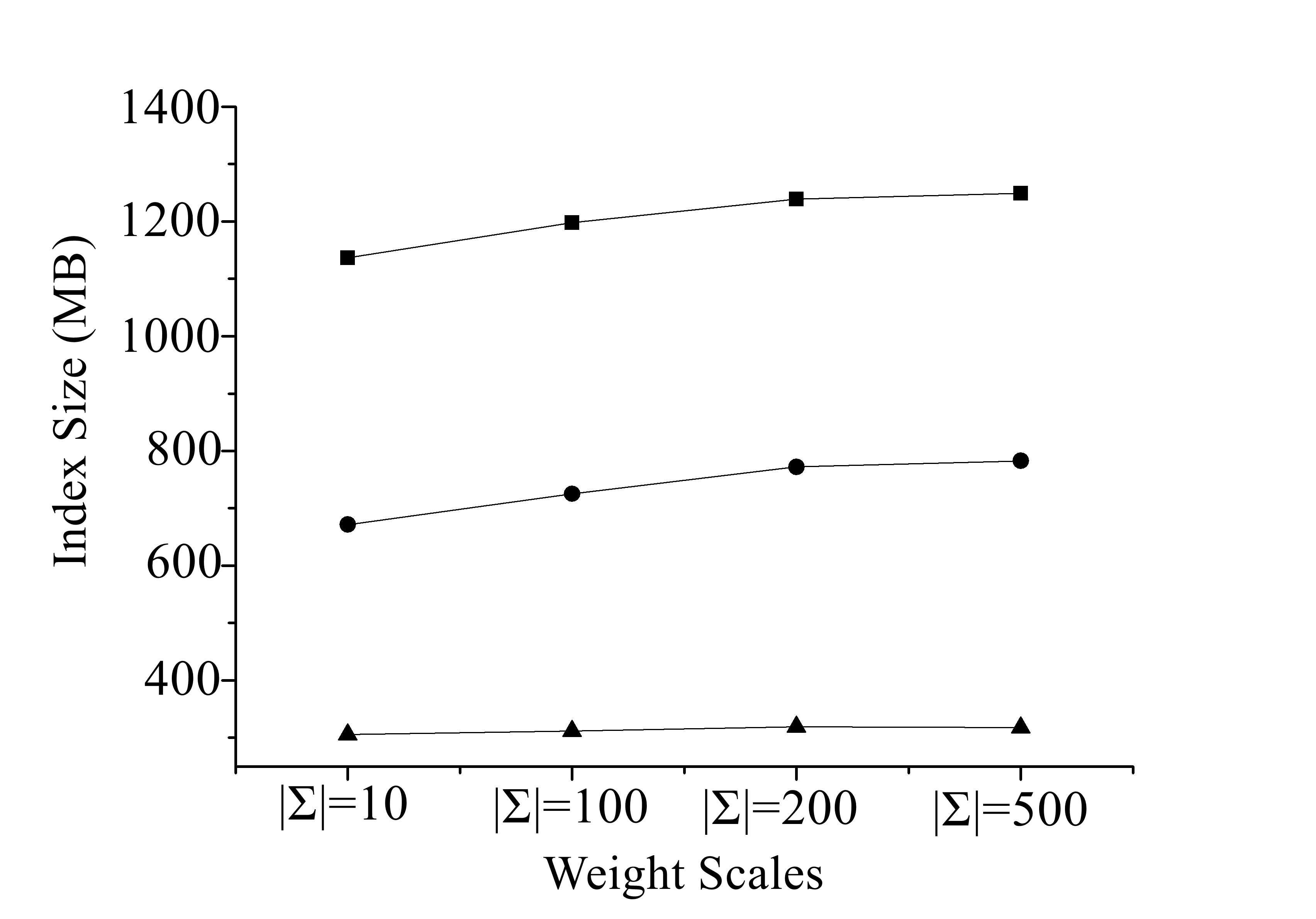}}
  \subfloat[cit-Patents]
  {\includegraphics[width=0.25\textwidth]{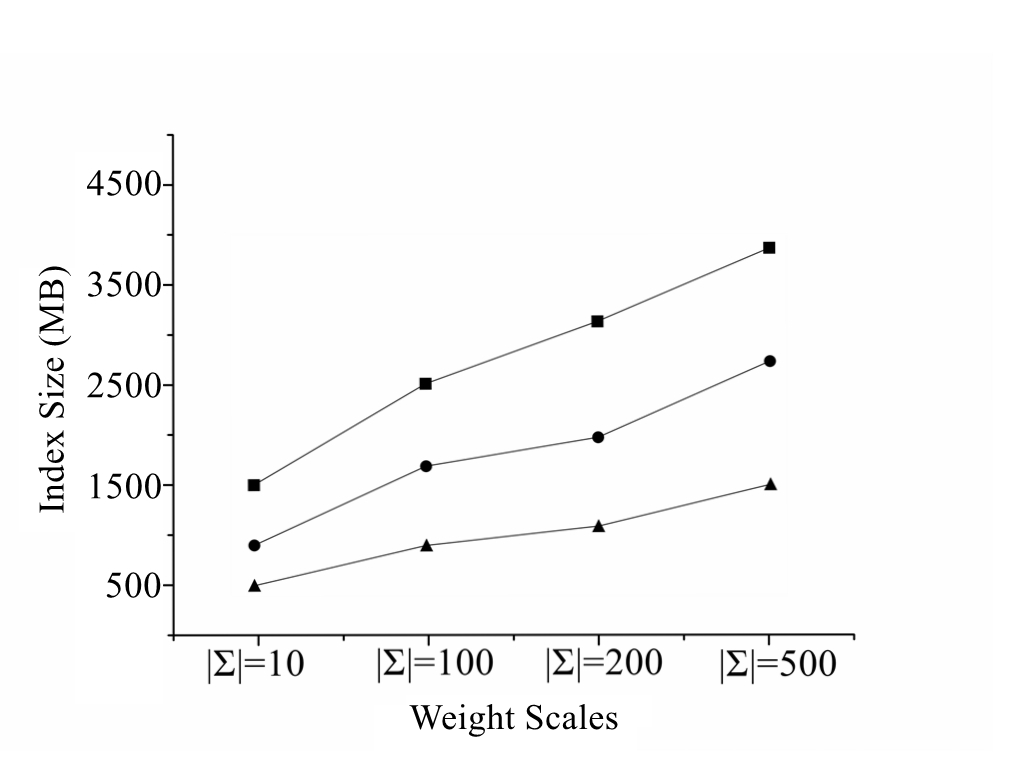}}
  \subfloat[10go-uniprot]
  {\includegraphics[width=0.25\textwidth]{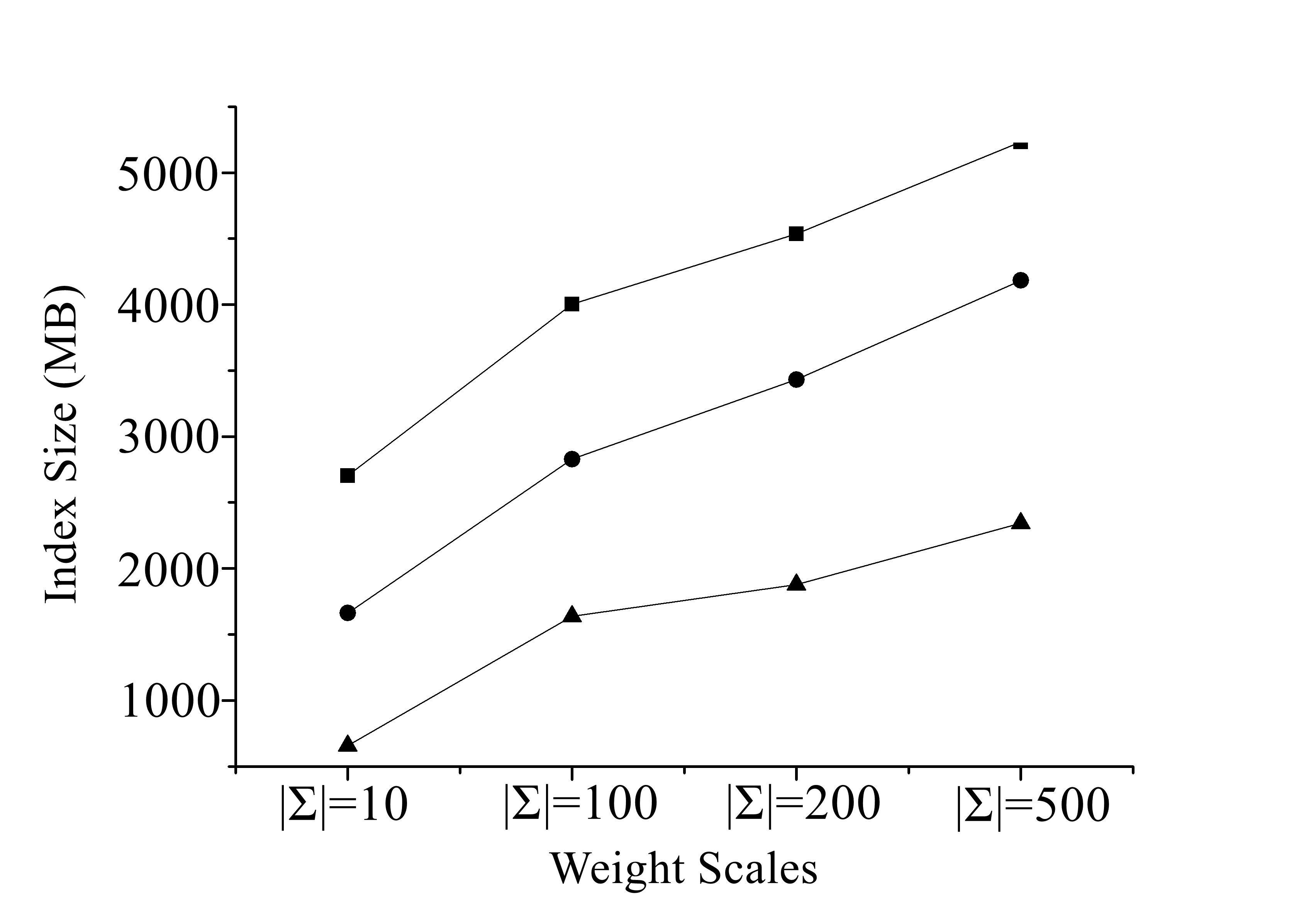}}
  \caption{Comparison of index size under different weight ranges(s)
}
\end{figure}
\begin{figure}[h]
\centering
  \subfloat[soc-LiveJournall]
  {\includegraphics[width=0.25\textwidth]{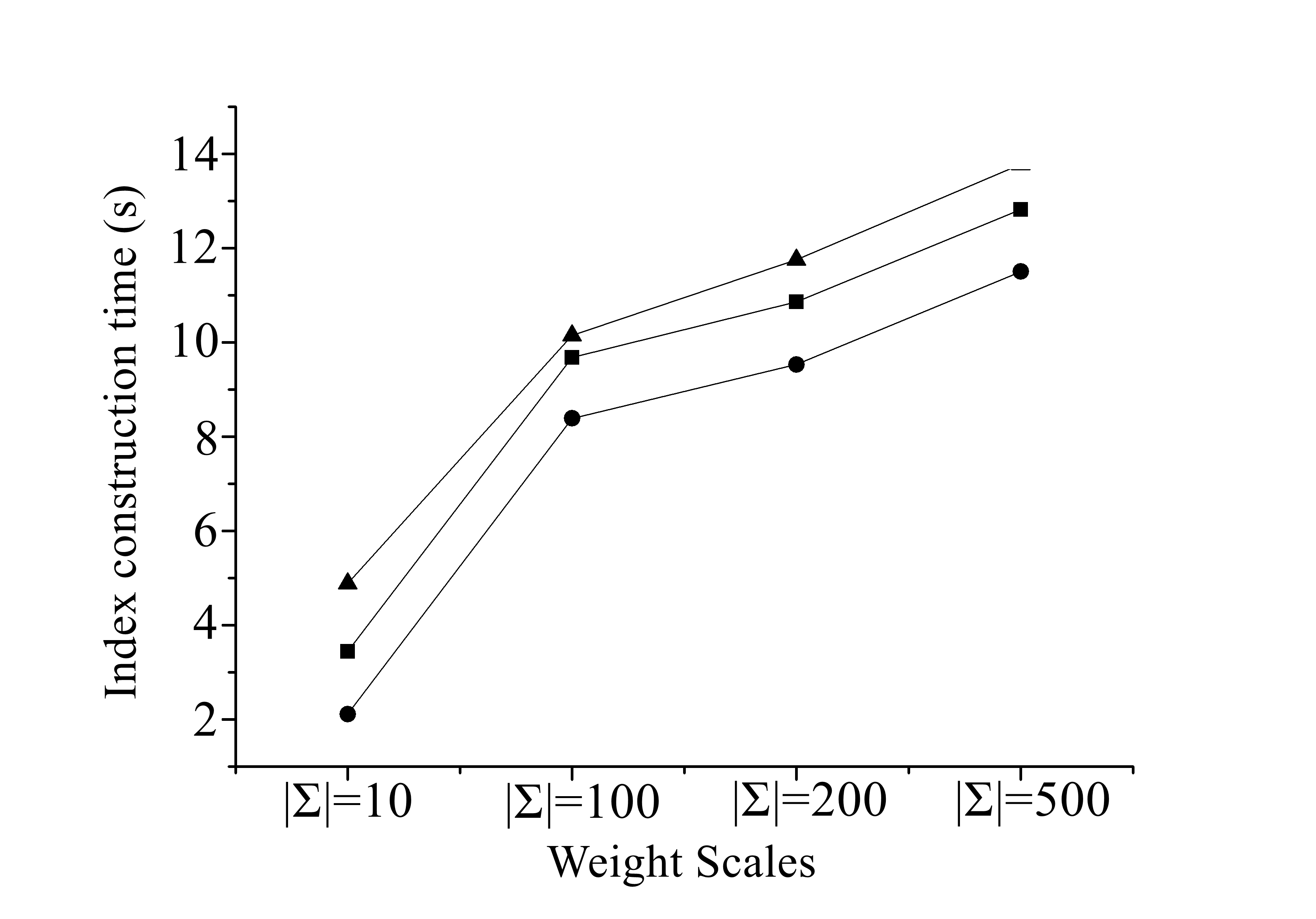}}    
  \subfloat[uniprot150m]
  {\includegraphics[width=0.25\textwidth]{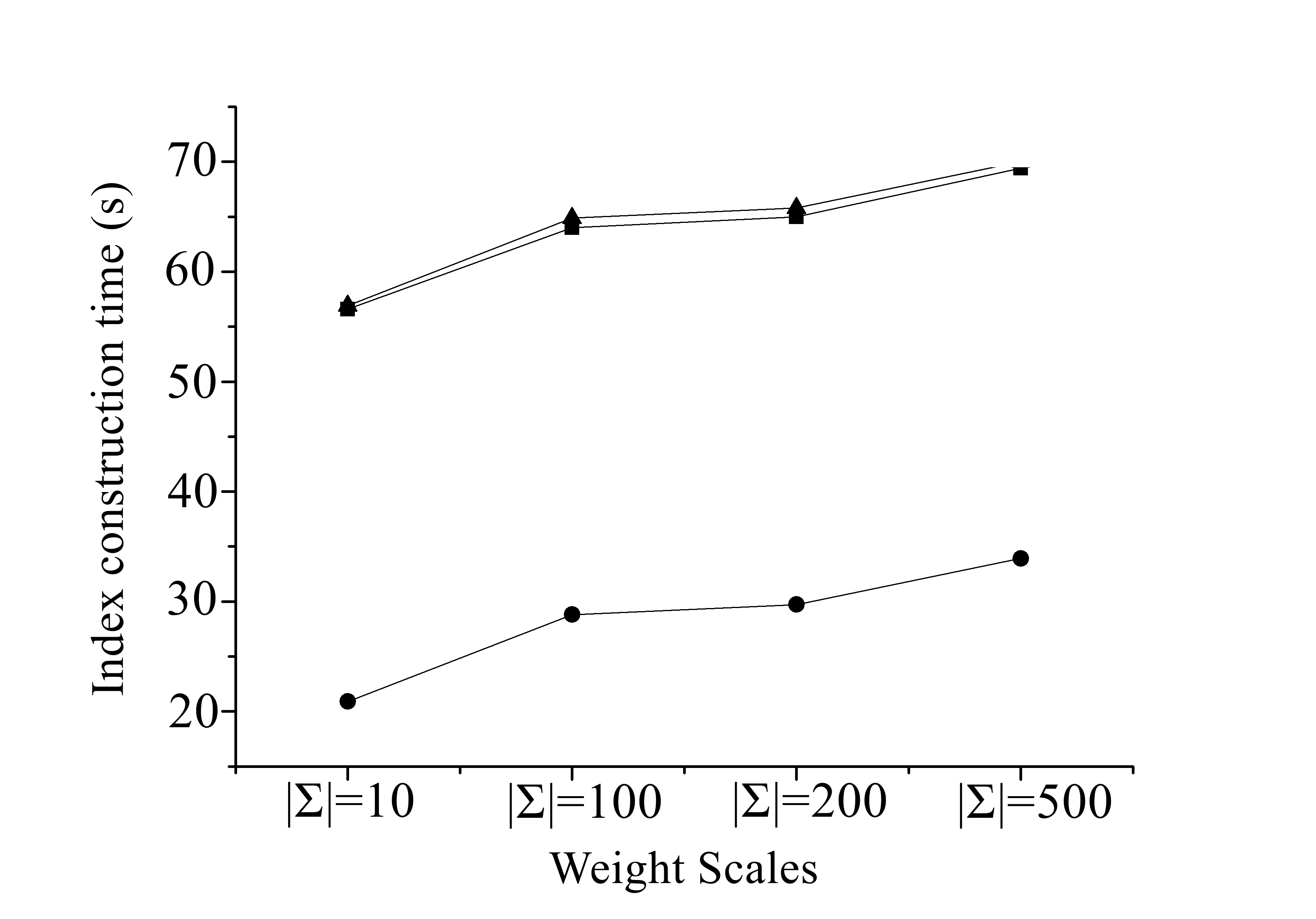}}
  \subfloat[cit-Patents]
  {\includegraphics[width=0.25\textwidth]{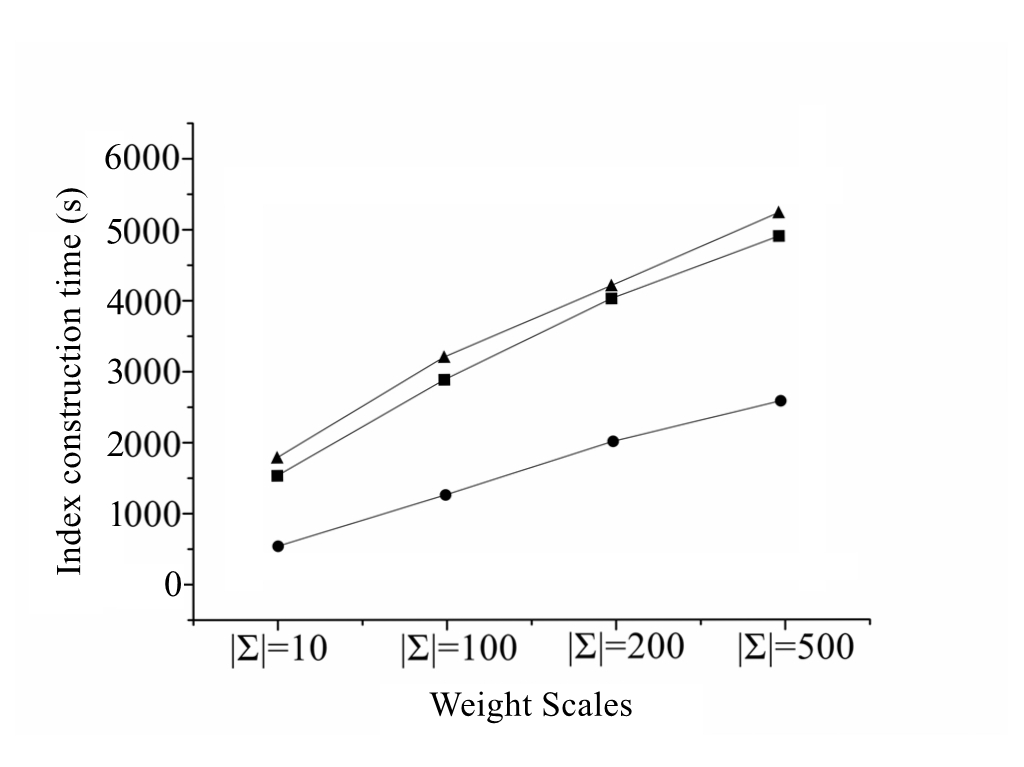}}
  \subfloat[10go-uniprot]
  {\includegraphics[width=0.25\textwidth]{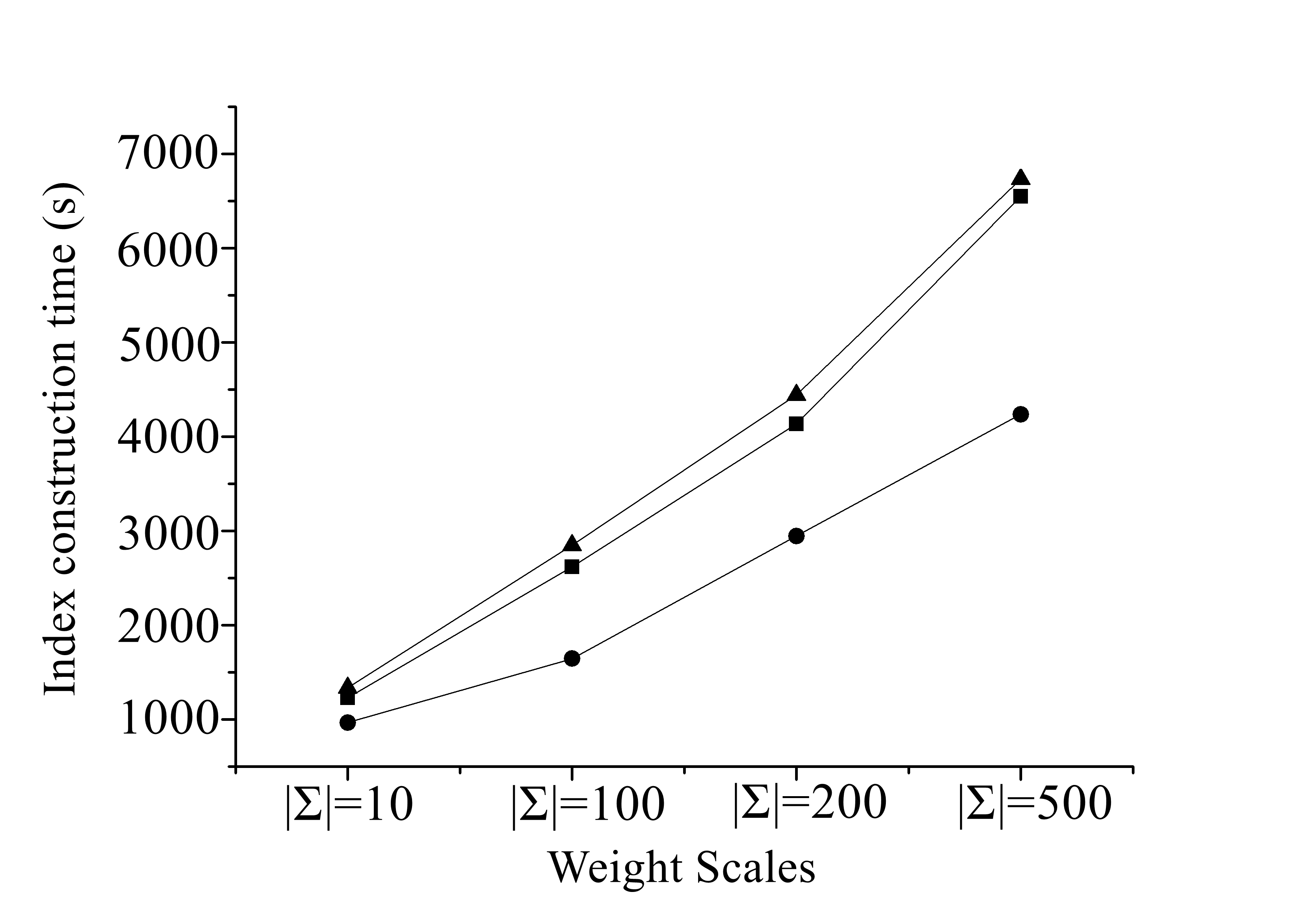}}
  \caption{Comparison of index construction time under different weight ranges(s)
}
\end{figure}
\begin{figure}[h]
\centering
  \subfloat[soc-LiveJournall]
  {\includegraphics[width=0.25\textwidth]{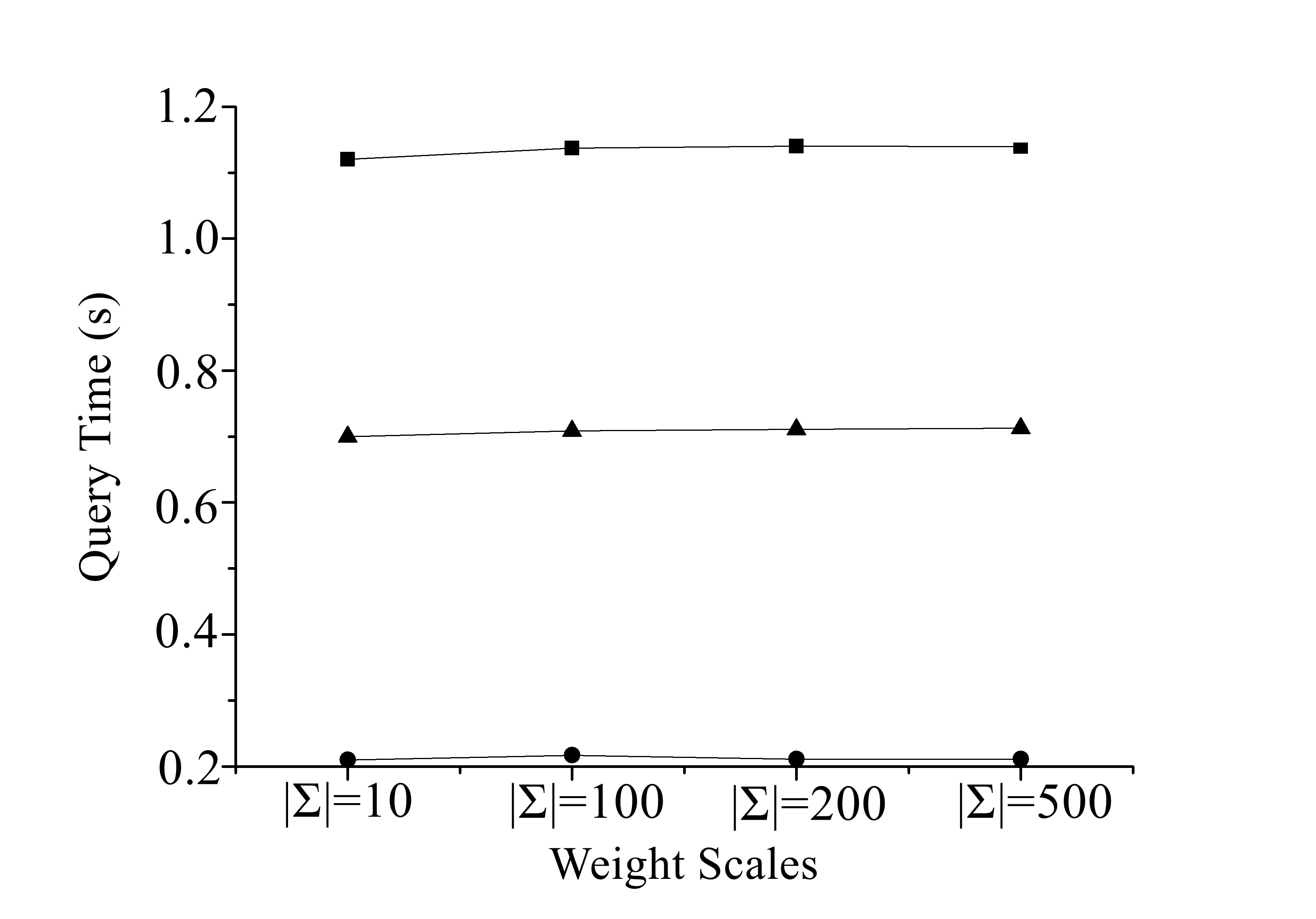}}    
  \subfloat[uniprot150m]
  {\includegraphics[width=0.25\textwidth]{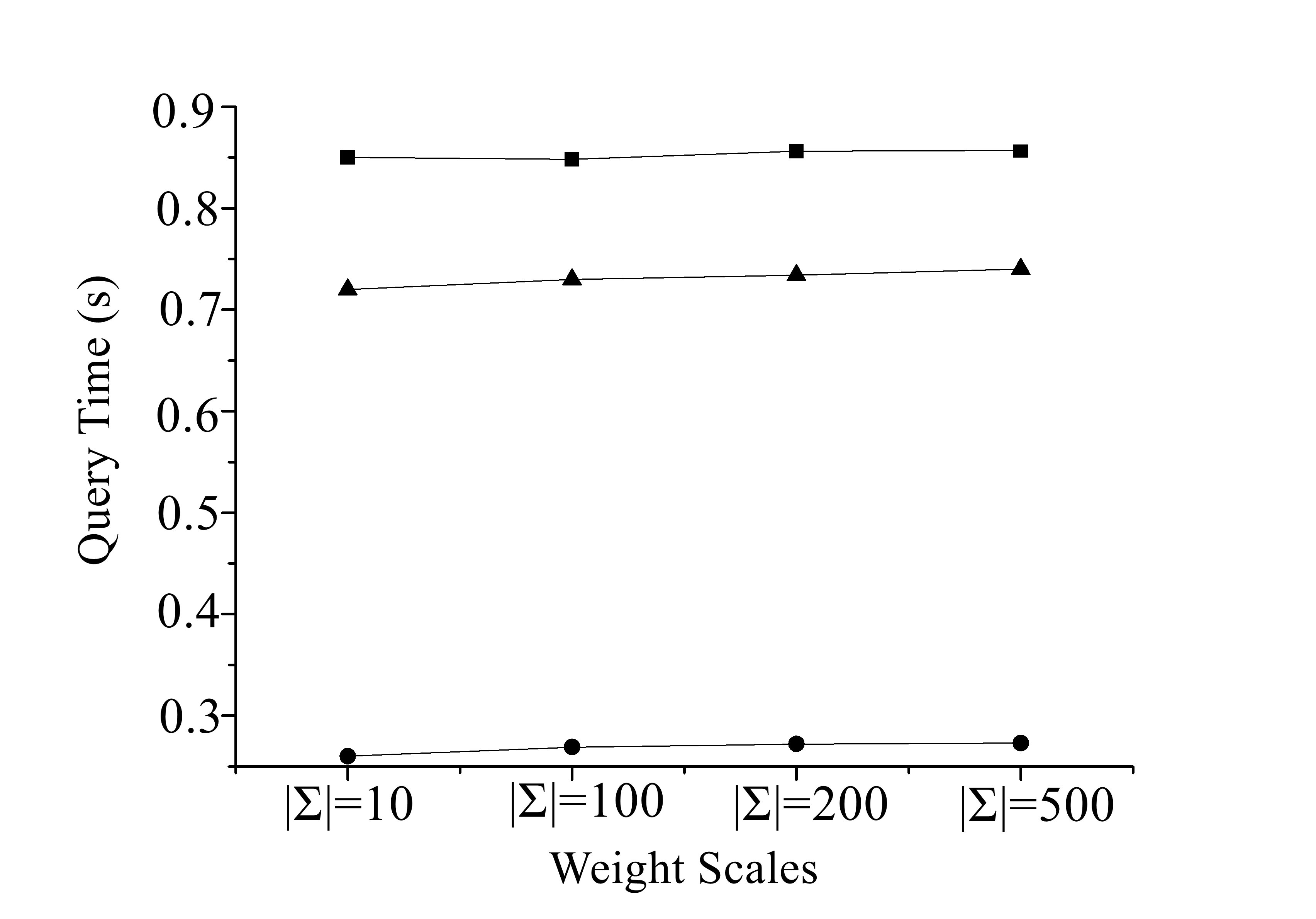}}
  \subfloat[cit-Patents]
  {\includegraphics[width=0.25\textwidth]{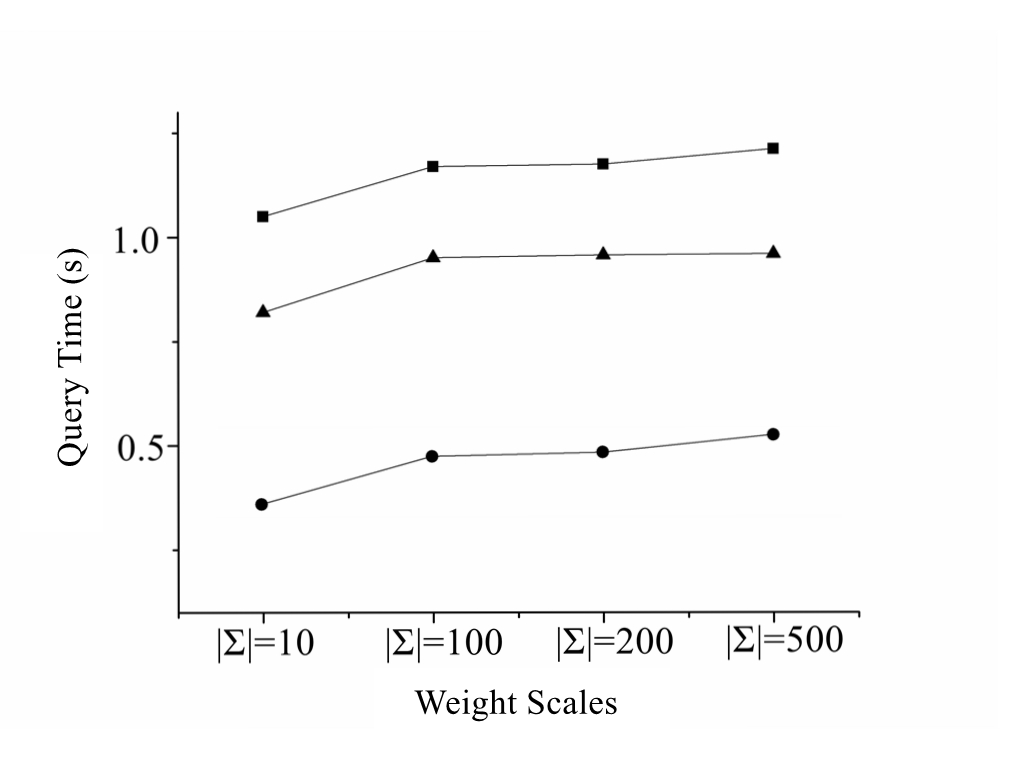}}
  \subfloat[10go-uniprot]
  {\includegraphics[width=0.25\textwidth]{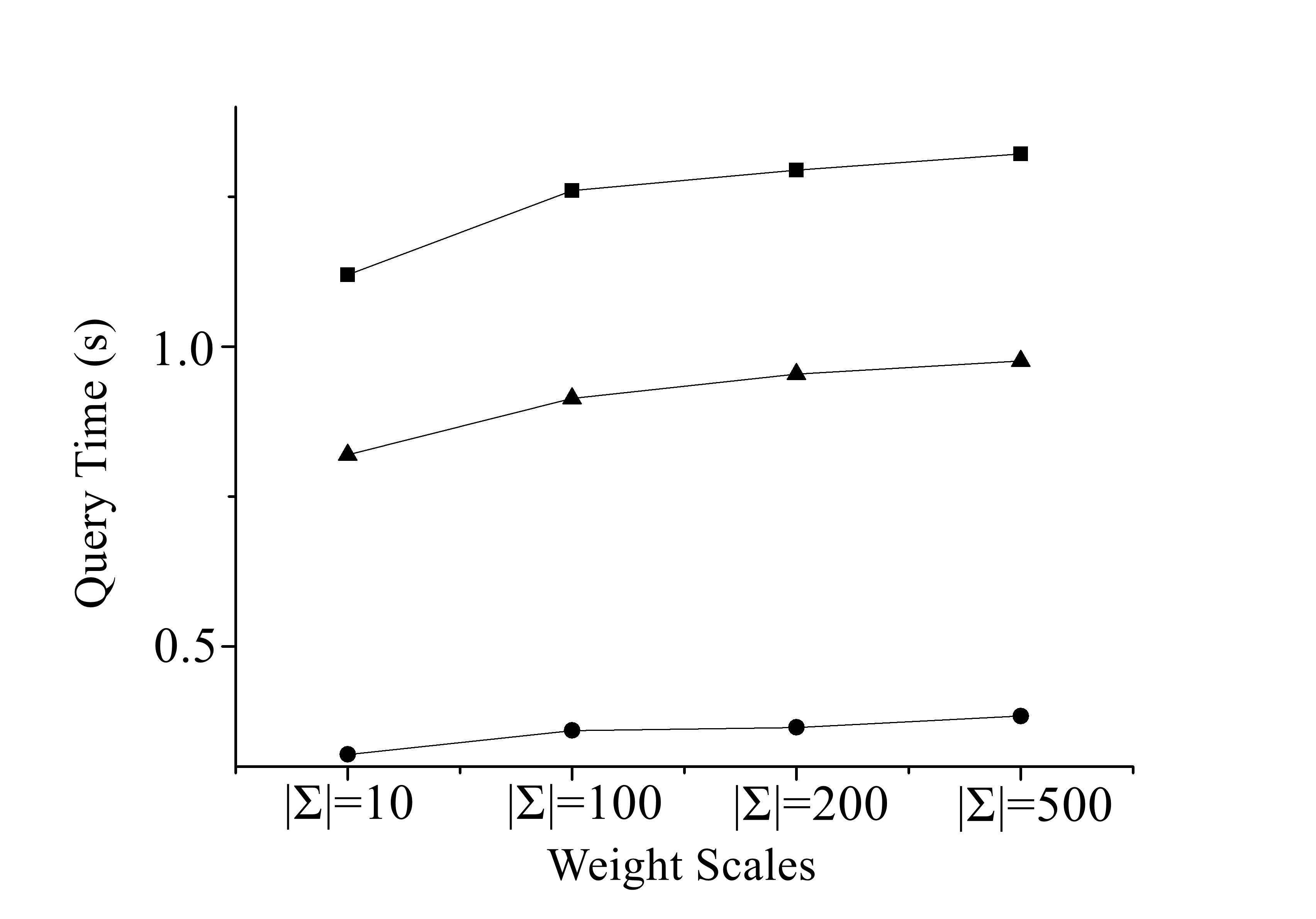}}
  \caption{Comparison of query time under different weight ranges(s)
}
\end{figure}

It can be seen from Figure 7, Figure 8 and Figure 9 that when the size of the weight set increases, the size of the index size and the construction time of the index generally show an increasing trend. This is because the weights are part of the index.  When there are more weights, the index size will also increase accordingly. At the same time, since the number of index edges that need to be traversed also increases, the index construction time will also increase at any time. When the size of the weight set increases, the query time increases slightly, but the overall difference is not much, and the query time is relatively stable.

In addition, by comparing the three algorithms, it can be found that due to the introduction of coverage sets to reduce the number of hop points, GWKRI and LWKRI have different optimization effects on the three indicators. LWKRI performs best if index size is considered, GWKRI performs best if index build time is considered, and GWKRI performs best if runtime is considered. The specific reasons have been explained in the previous experimental part and will not be repeated here.

\section{Conclusion}
In this paper, we have addressed the k-step reachability query problem with weight constraints by proposing the WKRI index and its associated construction algorithm. This approach enhances query efficiency and scalability over basic algorithms by leveraging path weight and distance information through a 2-hop indexing strategy. We further optimized this approach by introducing the GWKRI and LWKRI indexes, which build on the concept of minimum vertex cover to significantly improve index construction efficiency, reduce index size, and enhance query performance.
Experimental results demonstrate that the GWKRI index improves construction efficiency by a factor of 4, reduces index size by approximately 2 times, and increases query efficiency by about 6 times compared to the WKRI index. Additionally, the LWKRI index achieves a further reduction in index size by approximately 3 times compared to GWKRI, showcasing its potential for space-efficient indexing.
Despite these advancements, several limitations and avenues for future research remain. Firstly, the proposed indexes, while efficient, are still constrained by the inherent computational complexity of determining minimum vertex covers, particularly in large and dense graphs. Future work could explore approximation algorithms or heuristic methods to address this issue and further optimize index construction times.
Our current approach focuses on static graphs. Investigating dynamic indexing methods that efficiently handle graph updates, remains an open challenge. 
\bibliography{reference}

\end{document}